%% file: main.tex
\documentclass[11pt]{article}
 
\usepackage[
  top=1.1in,
  bottom=1.1in,
  left=1.1in,
  right=1.1in
]{geometry}

\usepackage{amsmath,amssymb,amsfonts}
\usepackage{amsthm}
\usepackage{mathtools}
\usepackage{bbm}        
\usepackage{mathrsfs}   
\usepackage{bm}         
 \usepackage{comment}
\usepackage{graphicx}
\usepackage{booktabs}
\usepackage{array}
\usepackage[numbers]{natbib}
\usepackage{algorithm}
\usepackage{algpseudocode}
\input{preamble}

\usepackage[dvipsnames]{xcolor}
\usepackage[
    linktoc=page,
    colorlinks=true,
    linkcolor=RoyalBlue,
    citecolor=BrickRed,
    urlcolor=black
]{hyperref} 

\theoremstyle{plain}
\newtheorem{theorem}{Theorem}[section]
\newtheorem{lemma}[theorem]{Lemma}

\newtheorem{corollary}[theorem]{Corollary}

\newtheorem{assumption}[theorem]{Assumption}
\newtheorem{remark}[theorem]{Remark}

\newcommand{\dist}{\mathrm{dist}}

\newcommand{\1}{\mathbbm{1}}

\newcommand{\Id}{\mathrm{Id}}

\newcommand{\calK}{\mathcal{K}}
\newcommand{\calS}{\mathcal{S}}
\newcommand{\calX}{\mathcal{X}} 
\newcommand{\calU}{\mathcal{U}} 
\newcommand{\calN}{\mathcal{N}} 
\newcommand{\calD}{\mathcal{D}} 
\newcommand{\calJ}{\mathcal{J}} 
\newcommand{\calG}{\mathcal{G}} 
\newcommand{\calI}{\mathcal{I}} 
\newcommand{\calV}{\mathcal{V}}

\newcommand{\sign}{\operatorname{sign}}
\allowdisplaybreaks 

\title{Uniform Recovery of Structured Signals
from Nonlinear Observations: Improved Error Rates}

\author{Pedro Abdalla\thanks{University of Wisconsin–Madison. (\texttt{pabdalla@wisc.edu})}\and Radu Balan\thanks{University of Maryland, College Park. (\texttt{rvbalan@umd.edu})} \and Junren Chen\thanks{Columbia University. (\texttt{jc6315@columbia.edu})}
}

\date{\today}

\begin{document}
\maketitle

\begin{abstract} 
Consider the recovery of structured signals, such as sparse vectors and low-rank matrices, from nonlinear observations. Under a Gaussian sensing matrix and a large class of unknown nonlinear link functions, Plan and Vershynin \cite{plan2016generalized} showed that generalized Lasso achieves accurate \emph{nonuniform} recovery of a fixed signal. More recently, Genzel and Stollenwerk \cite{genzel2023unified} showed that generalized Lasso is indeed capable of accurately recovering all structured signals. However, in some canonical settings with discontinuous link functions, their \emph{uniform} recovery error rate is  essentially slower than the nonuniform one. Specifically, in the recovery of $n$-dimensional  $k$-sparse vectors from $m$ measurements, generalized Lasso with a perfectly tuned $\ell_1$ constraint   achieves nonuniform error rate $ O(\sqrt{k\log(en/k)/m})$ \cite{plan2016generalized}, while the uniform error rate of  \cite{genzel2023unified} is no faster than $O((k\log(en/k)/m)^{1/4})$. In this paper, we narrow this gap by establishing improved uniform recovery guarantees under piecewise Lipschitz link functions with well-separated jump discontinuities. We analyze a projected gradient descent (PGD) algorithm whose projection can be onto a   convex set or a cone, and our results for the PGD with a convex set are also valid for the generalized Lasso.  In sparse recovery, the improved uniform error rates match the nonuniform rate $ O(\sqrt{k\log(en/k)/m})$ up to logarithmic factors. Under the sign link function, we further show that iterative hard thresholding (a specific instance of the PGD) achieves uniform recovery error rate $O(\sqrt{k\log(en/k)/m})$, matching the nonuniform rate up to a universal constant. Technically, the uniform guarantees for the PGD are obtained by showing that the gradient maps satisfy the restricted approximate invertibility condition uniformly over all signals. We demonstrate that this is a general approach to uniform recovery under nonlinear observations.  
\end{abstract}

\medskip


\smallskip


\medskip

\section{Introduction}
We consider the recovery of $\bx\in \calX$ from the observations 
\begin{align}
    y_i = f_i(\ba_i^\top \bx),\qquad i\in[m]:=\{1,...,m\}, \label{sim}
\end{align}
where $\ba_i$ are the sensing vectors, $\calX$ is a set of structured signals, and $f_i$ denote some nonlinear link functions. It is standard to study random sensing ensemble in theoretical analysis, under which an important feature of recovery guarantee is the (non)uniformity. A nonuniform guarantee for an estimator $\hat{\bx}$ ensures the accurate recovery of a fixed $\bx$ {\it oblivious} to $\bA$, taking the form \[\mathbb{P}(\|\hat{\bx}-\bx\|_2<\eta_1)\ge 1-\eta_2\] for some small $\eta_1,\eta_2>0$. In contrast, a uniform guarantee is stronger and states that $\hat{\bx}$ well approximates all $\bx$ in $\calX$:
\[\mathbb{P}\bigg(\sup_{\bx\in \calX}\|\hat{\bx}-\bx\|_2<\eta_1\bigg)\ge 1-\eta_2.\]

Uniformity is an important desideratum in mathematical analysis, since  the sensing ensemble in real-world  applications is fixed when designed and is expected to be able to work with all possible signals. Uniform guarantee is also more ``robust'' in that it allows for an adversarial generation of $\bx\in \calX$  that can be based on knowledge of $\bA$. From this perspective, uniformity may find interest  beyond signal processing, for example, in statistical learning.

For concreteness, we now consider the sparse recovery problem where $\calX\subset \Sigma^n_k:=\{\bx\in \mathbb{R}^n:\|\bx\|_0\le k\}$.
In the linear setting with $f_i = \Id$, that is the reconstruction of $\bx\in \Sigma^n_k$ from $\by=\bA\bx$, the compressed sensing theory provides a number of efficient algorithms that achieve uniform recovery of all $\bx
\in \Sigma^n_k$ under the restricted isometry property (RIP). See, e.g., \cite{Foucart2013AMI}. However, if $f_i$ is  a nonlinear link function, the problem of establishing uniform guarantee becomes much more entangled.

In this paper, we provide a case study of  the generalized Lasso \cite{plan2016generalized} and its corresponding projected gradient descent (PGD) algorithm \cite{oymak2017fast}, which are capable of accurately recovering the structured signal $\bx$ in a general Gaussian single-index model  with \emph{unknown} link function. Our goal is to understand their abilities to uniformly recover all $\bx\in\calX$, as with the recent work \cite{genzel2023unified}.

\subsection{Generalized Lasso and  Prior (Non)Uniform Guarantees}
In   parameter estimation from  (\ref{sim}) with   unknown link functions,   we  have to assume $\bx\in\mathbb{S}^{n-1}:=\{\bu\in \mathbb{R}^n:\|\bu\|_2=1\}$ since the signal norm information is absorbed by the unknown link function. Let $\|X\|_{\psi_2}=\inf\{t>0:\mathbb{E}[\exp(X^2/t^2)]<2\}$ be the $\psi_2$ norm of a random variable $X$. We make the following model assumptions. 
\begin{assumption}\label{assump:tsim}
    $\ba_1,\cdots,\ba_m$ are known i.i.d. $N(0,\bI_n)$ vectors. Independent of $\{\ba_i\}_{i=1}^m$, let $f_1,\cdots,f_m$ be i.i.d. copies of an unknown and random function $f:\mathbb{R}\to \mathbb{R}$ (including deterministic $f$ as a degenerate case). For some true signal $\bx\in\calX\subset \mathbb{S}^{n-1}$, one has access to the nonlinear observations $y_i=f_i(\ba_i^\top\bx)$, $i\in[m]$ as in (\ref{sim}). For some $\mu>0,\sigma>0$ and $g\sim N(0,1)$, the link function $f$ satisfies 
    \begin{gather}\label{mufgg}
        \mu = \mathbb{E}[f(g)g],\\\label{sigmafgmg}\sigma= \|f(g)-\mu g\|_{\psi_2}.
    \end{gather}
    Let $\bA=[\ba_1,\cdots,\ba_m]^\top$, we may abuse the notation to write $\bm{f}=(f_1,\cdots,f_m)^\top$ and 
    then express the observation vector as 
    \[\by = \bm{f}(\bA\bx):= \big(f_1(\ba_1^\top\bx),\cdots,f_m(\ba_m^\top\bx)\big)^\top.\]
\end{assumption}

 Under (\ref{mufgg}), the idea of generalized Lasso \cite{plan2016generalized} is that one can disregard the nonlinear link function and directly minimize the $\ell_2$ loss to obtain an accurate estimate of $\mu\bx$. Specifically, by choosing $\calK$ as a convex set that contains $\mu\bx$, the generalized Lasso is given by
\begin{align}
    \hat{\bx}_{\rm GL}= \textrm{arg}\min_{\bu\in \calK}\,\|\bA\bu-\by\|_2^2. \label{Glasso11}
\end{align}

Define $\mathbb{B}_p^n:=\{\bu \in \mathbb{R}^n:\|\bu\|_p\le 1\}$, and write $\mathbb{B}_2^n(\bv;r):=\{\bu \in \mathbb{R}^n: \|\bu-\bv\|_p\le r\}$.
Consider $\calU,\calV\subset \mathbb{R}^n$. Define \[\omega(\calU):=\mathbb{E}_{\bg\sim N(0,\bI_n)}\sup_{\bu\in\calU}\,\langle\bg,\bu\rangle\] as the Gaussian mean width of $\calU$. Under scale $\phi\ge 0$, we also define the local mean width as 
\[\omega_\phi(\calU):=\omega(\calU\cap \phi\mathbb{B}_2^n).\] Let $\calU\pm\calV:=\{\bu\pm\bv:\bu\in\calU,\bv\in\calV\}$ denote the Minkowski sum and difference of $\calU$ and $\calV$, and also let $\lambda \calU=\{\lambda\bu:\bu\in\calU\}$ for $\lambda\in \mathbb{R}$. Throughout this paper, we use $C,C_i,c,c_i$ to denote universal constants whose values may change from line to line. We write $T_1\gtrsim T_2$ (or $T_1=\Omega(T_2)$) if $T_1\ge cT_2$ for some $c>0$, and write $T_1\lesssim T_2$ (or $T_1=O(T_2)$) if $T_1\le CT_2$. In the following, we revisit the nonuniform recovery guarantee of \cite{plan2016generalized}. 
\begin{theorem}[Nonuniform recovery; {\cite[Theorem 1.9]{plan2016generalized}}] \label{thm:pvnonu}Under Assumption \ref{assump:tsim}, assume that the  fixed $\bx$ satisfies $\mu\bx\in\calK$ where $\calK$ is convex. For any $\phi>0$, if 
\begin{align}
    m\gtrsim \omega^2_1\Big(\frac{\calK-\mu\bx}{\phi}\Big) \gtrsim 1,
\end{align}
then with high probability, $\hat{\bx}_{\rm GL}$ defined in (\ref{Glasso11}) satisfies
\begin{align}
    \|\hat{\bx}_{\rm GL} -\mu\bx\|_2 \lesssim  \frac{\sigma}{\sqrt{m}}\omega_1\Big(\frac{\calK-\mu\bx}{\phi}\Big) + \phi.
\end{align}
\end{theorem}

\begin{remark}
    For this result to hold, the $\psi_2$-norm condition in (\ref{sigmafgmg}) can be relaxed to $\max\{\mathbb{E}[(f(g)-\mu g)^2],\mathbb{E}[(f(g)-\mu g)^2g^2]\}\lesssim \sigma^2$ for $g\sim N(0,1)$. We simply work with (\ref{sigmafgmg}) in this paper since it is adopted in the previous uniform recovery guarantees of \cite{genzel2023unified}. 
\end{remark}

Next, we revisit the uniform recovery guarantees established in \cite{genzel2023unified} under the corruption-free setting in Assumption \ref{assump:tsim}. Their guarantees   exhibit essential discrepancy under Lipschitz continuous link functions and discontinuous link functions. We start with the Lipschitz links. 

\begin{theorem}[Uniform recovery under Lipschitz links; {\cite[Theorem 1]{genzel2023unified}}] \label{thm:gsconu}Under Assumption \ref{assump:tsim}, assume that $\mu\calX \subset \calK$ where $\calK$ is convex, and that (any realization of) $f:\mathbb{R}\to \mathbb{R}$ is $\gamma$-Lipschitz. There exist universal constants $C>c>0$ such that the following holds. For any $\phi>0$, if 
\begin{align}
    m \gtrsim \omega^2_1\Big(\frac{\calK-\mu\calX}{\phi}\Big) \gtrsim 1,
\end{align}
then with probability at least $1-C\exp\big(-c\omega^2_{1}(\frac{\calK-\mu\bx}{\phi})\big)$, $\hat{\bx}_{\rm GL}$ defined in (\ref{Glasso11}) satisfies 
\begin{align}
     \sup_{\bx\in\calX}\|\hat{\bx}_{\rm GL}-\mu\bx\|_2\lesssim  \frac{\sigma}{\sqrt{m}} \omega_1\Big(\frac{\calK-\mu\calX}{\phi}\Big) + \phi + \frac{(\mu+\gamma)\omega(\calX)}{\sqrt{m}}.
\end{align}
\end{theorem}

The uniform recovery guarantee under discontinuous $f$ is more intricate. The high-level idea is to construct Lipschitz continuous $\{f_\phi:\phi>0\}$   to approximate $f$ such that $\varepsilon^*_\phi:=|f-f_\phi|$ is also Lipschitz continuous. Note that \cite{chen2023unified}, which extended the idea to signal recovery under generative priors, showed that such $f_{\phi}$ can be constructed for piecewise Lipschitz $f$ with separated bounded discontinuities  \cite[Assumption 2, Section 3.1]{chen2023unified}.

\begin{theorem}[Uniform recovery under discontinuous links; {\cite[Theorem 2]{genzel2023unified}}] \label{thm:gsdiscu}
Under Assumption \ref{assump:tsim} (without the need of (\ref{sigmafgmg})), assume that $\mu\calX \subset \calK$ where $\calK$ is convex. Given $\phi>0$, assume that for any realization of $f:\mathbb{R}\to \mathbb{R}$ there is a corresponding $f_\phi:\mathbb{R}\to \mathbb{R}$ such that the following holds: let $\varepsilon^*_{\phi}(b):=|f(b)-f_{\phi}(b)|$ and $\xi^*_{\phi}(b):=\mu b - f_\phi(b)$ for $b\in \mathbb{R}$, and $(g_1,g_2)$ be any zero-mean random Gaussian vector of unit variances (i.e., $\mathbb{E}[g_1^2]=\mathbb{E}[g_2^2]=1$), 
\begin{gather}
    \mathbb{E}|\varepsilon^*_{\phi}(g_1)g_2| \le \frac{\phi}{64},\qquad \|\xi_\phi^*(g_1)\|_{\psi_2}\le \sigma_1,\qquad\|\varepsilon^*_{\phi}(g_1)\|_{\psi_2}\le \sigma_2 ;\\ \label{increment1}
    \|\varepsilon_{\phi}^*(\ba_i^\top \bu)-\varepsilon_{\phi}^*(\ba_i^\top \bv)\|_{\psi_2}\le L_{1,\phi}\|\bu-\bv\|_2,\qquad \forall \bu,\bv\in \calX,\\\label{increment2}
    \|\xi_{\phi}^*(\ba_i^\top \bu)-\xi_{\phi}^*(\ba_i^\top \bv)\|_{\psi_2}\le L_{2,\phi}\|\bu-\bv\|_2,\qquad \forall \bu,\bv\in \calX. 
\end{gather}
If 
\begin{align}\label{disucon}
    m\gtrsim \omega^2_{1}\bigg(\frac{\calK-\mu\calX}{\phi}\bigg) + \frac{(L_{1,\phi}^2+L_{2,\phi}^2)\mu^2\omega^2(\calX)}{\phi^2}\gtrsim 1,
\end{align}
then with probability at least  $1-C\exp\big(-c\omega^2_{1}(\frac{\calK-\mu\bx}{\phi})\big)$, $\hat{\bx}_{\rm GL}$ defined in (\ref{Glasso11}) satisfies 
\begin{align}
     \sup_{\bx\in\calX}\|\hat{\bx}_{\rm GL}-\mu\bx\|_2\lesssim  \frac{\sigma_1+\sigma_2}{\sqrt{m}} \omega_1\Big(\frac{\calK-\mu\calX}{\phi}\Big)+ \phi.
\end{align}
\end{theorem}

\subsection{A Gap under Discontinuous Links}\label{sec:gapidentify}
In the prior guarantees, there is a gap between uniform recovery and nonuniform recovery under discontinuous links. To exhibit this gap, we let $\calD(\calK,\mu\calX):=\{\lambda\bh:\lambda\ge 0,\bh\in \calK-\mu\calX\}$ be the descent cone of the pair $(\calK,\mu\calX)$ and discuss the following canonical sparse recovery setting. Throughout the paper we write $\Sigma^{n,*}_k =\Sigma^{n}_k \cap \mathbb{S}^{n-1}$.

\begin{assumption}[Sparse recovery with perfectly tuned $\ell_1$ constraint] \label{assump:sparsecone}For some $R>0$ and sparsity $k\in[n]$, $\calX= \Sigma^{n,*}_k \cap  \{\bu\in\mathbb{R}^n:\|\bu\|_1=R\}$, $\calK=\mu R\mathbb{B}_1^n$. In this case, for any $\phi>0$ one has
\[\omega_{1}\Big(\frac{\calK-\mu\calX}{\phi}\Big)\le \omega_1(\calD(\calK,\mu\calX))\lesssim \sqrt{k\log\frac{en}{k}};\]
see, e.g., \cite[Equation (3.3)]{genzel2023unified}. Also, $\omega(\calX)\lesssim \omega(\Sigma^{n,*}_k)\lesssim\sqrt{k\log\frac{en}{k}}$; see, e.g., \cite{plan2012robust}. 
\end{assumption}

Under Assumptions \ref{assump:tsim} and \ref{assump:sparsecone},  by setting $\phi= \sigma\sqrt{k/m}$, the nonuniform rate in Theorem \ref{thm:pvnonu} specializes to $O(\sigma\sqrt{k\log(en/k)/m})$ whenever $m\gtrsim k\log\frac{en}{k}$. Under $\gamma$-Lipschitz links, the uniform error rate in Theorem \ref{thm:gsconu}---again under $\phi=\sigma\sqrt{k/m}$---specializes to $O([\sigma+\mu+\gamma]\sqrt{k\log(en/k)/m})$, which matches the nonuniform rate up to the leading constant.

However, under discontinuous links, the uniform rate implied by Theorem \ref{thm:gsdiscu}---which is the only known uniform guarantee for Gaussian single-index model with unknown link, to our best knowledge---is in general essentially slower. The problem arises from the condition 
\begin{align}
    m\gtrsim \frac{(L^2_{1,\phi}+L^2_{2,\phi})\mu^2\omega^2(\calX)}{\phi^{2}}\label{disconucon}
\end{align}
required in (\ref{disucon}), since $L_{1,\phi}$ and $L_{2,\phi}$ typically worsen under smaller $\phi$. Specifically, in the applications to concrete discontinuous links in \cite[Section 4]{genzel2023unified}, including the one-bit/multi-bit quantizers and the modulo nonlinearity (i.e., $f(b)=m_{\lambda}(b):=b- \big\lfloor \frac{b+\lambda}{2\lambda}\big\rfloor 2\lambda$), the authors are only able to construct $f_{\phi}(b)$ to render $L_{1,\phi},L_{2,\phi}\asymp \phi^{-1}$; see Figure 2 therein for  illustration. Therefore, (\ref{disconucon}) reduces to $m\gtrsim \frac{\mu^2k\log(en/k)}{\phi^4}$, and consequently Theorem \ref{thm:gsdiscu} does not yield a uniform error rate faster than $O\big((\mu^2k\log(en/k)/m)^{1/4}\big)$. As such, the authors concluded ``{\it the transition to uniform recovery with nonlinear output functions (particularly discontinuous ones) may result in a worse oversampling rate}''   \cite[Page 916]{genzel2023unified}.

\begin{remark}
    In general, the gap remains visible if   $\calD(\calK,\mu\bx)=\{\lambda\bh:\lambda\ge 0,\bh\in \calK-\mu\calX\}$ has low Gaussian width; see, e.g., \cite[Theorem 1.4]{plan2016generalized}. Beyond sparsity, another canonical example is the recovery of low-rank matrices with $\calK$ being a perfectly tuned nuclear norm constraint. 
\end{remark}
\begin{remark}
    In \cite[Section 5]{genzel2023unified}, the authors demonstrate that existing random embedding results can be incorporated to avoid the increment conditions in (\ref{increment1})--(\ref{increment2}), and they follow this route to close the above gap under (dithered) one-bit measurements. However, such suitable random embedding result is only available for several specific links. In particular,   such embedding result does not exist for modulo measurements treated in \cite[Section 4.4]{genzel2023unified}, and consequently \cite{genzel2023unified} is not able to improve the $O(m^{-1/4})$ rate in Corollary 5 therein to $\tilde{O}(m^{-1/2})$ under a descent cone structure. As we shall see, the present paper achieves this improvement under a general class of discontinuous links (including the modulo link).
\end{remark}

\subsection{Our Contributions}
The main aim of this paper is to narrow the above gap, showing that the uniform error rate degrades from the nonuniform rate only by some \emph{logarithmic factors}, without essential rate degradation. We study the projected gradient descent (PGD) algorithm: start from $\bx_0\in \calK$ and run 
\begin{align} \label{pgd11}
    \bx_{t+1} = P_{\calK}\bigg(\bx_{t} - \frac{1}{m}\bA^\top(\bA\bx_t-\by)\bigg),\quad t=0,1,\cdots.
\end{align}
to obtain $\{\bx_t\}_{t\ge 0}$. Here, $P_{\calK}$ denotes the projection onto $\calK$ under $\ell_2$ norm. \footnote{Throughout the paper, we assume that $\calK$ is closed to ensure that $P_{\calK}$ is well defined.}
This can be viewed as a computational procedure for solving the generalized Lasso (\ref{Glasso11}), since $\frac{1}{m}\bA^\top(\bA\bu-\by)$ is just the gradient of the $\ell_2$ loss $\frac{1}{2m}\|\bA\bu-\by\|_2^2$. In fact, \emph{nonuniform} recovery guarantee for (\ref{pgd11}) that is comparable to Theorem \ref{thm:pvnonu} has been established  in \cite{oymak2017fast}. We will also show that our uniform recovery guarantee for a convex set $\calK$ remains true for generalized Lasso.

Throughout the paper, we use $\tilde{O}(\cdot)$ and $\tilde{\Omega}(\cdot)$ to hide the logarithmic factors in $(m,n)$ in $O(\cdot)$ and $\Omega(\cdot)$, respectively.  

\begin{theorem}[Informal version of Corollary \ref{cor:sparseconvex}] Under Assumptions \ref{assump:tsim} and \ref{assump:sparsecone} and piecewise Lipschitz links with separated jump discontinuities (precised in Assumption \ref{assump:fi}), with high probability, the PGD in (\ref{pgd11}) and the generalized Lasso in (\ref{Glasso11}) uniformly recover all $\bx\in\calX$ to $\tilde{O}(\sqrt{k/m})$ $\ell_2$ error. 
\end{theorem}

While it is important to use convex $\calK$ to render the generalized Lasso computationally feasible, this need not be the case    for PGD (\ref{pgd11}). In fact, (\ref{pgd11}) is computationally feasible so long as the projection $P_{\calK}$ can be efficiently computed. This is true for some highly nonconvex $\calK$, such as the cone of sparse vectors, for which $P_{\calK}$ is a hard thresholding operator, and the cone of low-rank matrices, for which $P_{\calK}$ is a truncated SVD. To be concrete, we consider the following sparse recovery setting.

\begin{assumption}
    [Sparse recovery via iterative hard thresholding] \label{assump:iht}For some sparsity $k\in[n]$, $\calX=\Sigma^{n,*}_k$ and $\calK=\Sigma^n_k$. In this case, $P_{\calK}$ reduces to the hard thresholding operator $H_k:\mathbb{R}^n\to \mathbb{R}^n$ that retains only the $k$ largest entries (in magnitude) and set the others to zero. 
\end{assumption}

 While Assumption \ref{assump:sparsecone} requires a tuning parameter $\mu  R$ which depends on the unknown link $f$ through $\mu$, we note that Assumption \ref{assump:iht} only requires the sparsity $k$ and therefore requires no knowledge from $f.$ 
We establish similar improved uniform error rate in the setting of Assumption \ref{assump:iht}.

\begin{theorem}[Informal version of Corollary \ref{cor:sparsecone}] Under Assumptions \ref{assump:tsim} and \ref{assump:iht} and piecewise Lipschitz links with separated jump discontinuities (precised in Assumption \ref{assump:fi}), with high probability, the PGD in (\ref{pgd11}) uniformly recovers all $\bx\in\calX$ to $\tilde{O}(\sqrt{k/m})$ $\ell_2$ error. 
\end{theorem}

\paragraph{Outline.} The remainder of this paper is organized as follows. In Section \ref{sec:sim} we present our main results, put them in perspective, and outline the proofs. Sections \ref{sec:convexproof}--\ref{sec:proveconepgd} give the complete proofs to Theorems \ref{thm:recoveryconvex} and \ref{thm:recoverycone}. In Section \ref{sec:discuss}, concluding remarks are given to close the paper. To keep the presentation flow, some of the detailed arguments and the technical lemmas are deferred to the appendices.

\section{Main Results}\label{sec:sim}
In this section, we provide our main results, position them in the literature, and provide a high-level technical overview.  Recall that  the PGD (\ref{pgd11}) leverages the gradient of the $\ell_2$ loss $L(\bu)=\frac{1}{2m}\|\bA\bu-\by\|_2^2$, i.e., $\nabla L(\bu)=\frac{1}{m}\bA^\top(\bA\bu-\by)$. Substituting $\by=\bm{f}(\bA\bx)$, we introduce a shorthand for the gradient:
\begin{align}
    \bh_{\bx}(\bu):=\frac{1}{m}\bA^\top(\bA\bu-\bm{f}(\bA\bx)).\label{hxshort}
\end{align}

We say a function $h:\mathbb{R}\to \mathbb{R}$ has a jump discontinuity at $x_0$ if $L^-:=\lim_{x\to x_0^-}h(x)$ and $L^+:= \lim_{x\to x_0^+}h(x)$ exist but are not equal. We simply call the oscillation at $x_0$, that is $|L^+-L^-|$, the jump. 
We also write $\calD_h:=\{a\in \mathbb{R}:h\textrm{ is discontinuous at $a$}\}$. Our uniform recovery guarantees apply to discontinuous links which are piecewise Lipschitz and only have jump discontinuities with bounded jumps and bounded away from each other. We make this precise in the following. 

\begin{assumption} \label{assump:fi}
Let $f$ be the function in Assumption \ref{assump:tsim}. There exist some constants $\gamma,\tau \in [0,\infty),~\beta\in(0,\infty]$ such that,   any realization of $f$ satisfies the following:
\begin{enumerate} 
    \item[(C1)] ({\it Piecewise Lipschitz}) $f$ is $\gamma$-Lipschitz over $(b_1,b_2)$ for any $-\infty\le b_1<b_2\le \infty$ such that $(b_1,b_2)\cap \calD_{f}=\varnothing.$ 
    \item[(C2)] ({\it Separation of discontinuities}) If $|\calD_{f}|> 1$,
    \(|\xi_1-\xi_2|\ge\beta\) holds for any \(\xi_1,\xi_2 \in \calD_{f}\) such that \(\xi_1\ne \xi_2.\)  
    (If $|\calD_{f}|\le 1$ holds for any realization of $f$, we let $\beta:=\infty$ and follow the convention $\frac{b}{\infty} = 0~(\forall b\in \mathbb{R})$.) 
    \item[(C3)] ({\it Bounded jumps}) Every discontinuity of $f$, if any, is a jump discontinuity with   jump bounded by $\tau$.  
    Also, assume $f(b)\in\{\lim_{c\to b ^-}f(c),\lim_{c\to b ^+}f(c)\}$ for $b\in \calD_f$. 
\end{enumerate}
\end{assumption} 
\begin{lemma}\label{lem:smallball}
    Under Assumption \ref{assump:fi}, for any realization of $f$ we let $\dist(x,\calD_f):=\min_{b\in \calD_f}|x-b|$ with the convention that $\dist(x,\varnothing)=\infty$. Then on the randomness of $g\sim N(0,1)$, for every $t\ge 0$ we have  
    \[\mathbb{P}\Big(\dist(g,\calD_f)\le t\Big)\le \Big(\frac{2}{\beta}+\sqrt{\frac{2}{\pi}}\Big)t.\]
\end{lemma}
\begin{proof}
The proof consists of some elementary calculations and is deferred to Appendix \ref{app:prooflem22}.
\end{proof}

\subsection{PGD with Convex $\calK$ and Generalized Lasso}

Beyond Gaussian width, our uniform recovery guarantees involve another complexity measure, that is the covering number. Under $\ell_2$ distance, the radius-$r$ covering number of $\calU\subset \mathbb{R}^n$ is defined as 
\[\calN(\calU;r):=\inf\bigg\{|\calS|:\calS\subset \calU\subset \bigcup_{\bs\in\calS}\mathbb{B}_2^n(\bs;r)\bigg\}.\]
We note the estimate (see, e.g., \cite[Lemma 3.3]{plan2013one})
\[\log\calN(\Sigma^{n,*}_k;r)\le k\log\frac{9n}{kr},\qquad \forall r\in (0,1).\]
\begin{theorem}[PGD and generalized Lasso with convex $\calK$]\label{thm:recoveryconvex}
Under Assumptions \ref{assump:tsim} and \ref{assump:fi}, assume that $\calK$ is a convex set and $\mu\calX\subset \calK$, then there exist universal constants $C>c>0$ such that the following holds. For any $(\phi,r,t,\zeta)\in (0,\infty)\times (0,1)\times [1,\infty) \times \{\frac{i}{m}:i\in [m]\}$ obeying $t^2\le \zeta m \log(e/\zeta)$ and 
   \[ \delta:=\zeta+\frac{\log\calN(\calX;r)+t^2}{m}
    +\bigg[\frac{2}{\beta}+\sqrt{\frac{2}{\pi}}\bigg]\bigg(
        \frac{\omega_r(\calX-\calX)}{\sqrt{\zeta m}}
        +r\sqrt{\log\frac{e}{\zeta}}
    \bigg)
    \le c,\]
    let
    \[\mathsf{B}_1(\phi,r,t,\zeta):=\frac{\sigma\big(\omega_1(\frac{\calK-\mu\calX}{\phi})+\sqrt{\log\calN(\calX;r)+t^2}\big)}{\sqrt m}+ (\mu+\gamma )r + \tau\sqrt{\delta}
    \bigg(\frac{\omega_1(\frac{\calK-\mu\calX}{\phi})}{\sqrt m}
        +\sqrt{\delta\log\frac{e}{\delta}}\bigg),\]
    if 
    \[m\ge C\bigg[ \omega_1^2\bigg(\frac{\calK-\mu\calX}{\phi}\bigg)+ \omega_1^2\bigg(\frac{\calX-\calX}{r}\bigg)\bigg],\]
    then with probability at least $1-C\exp(-ct^2)-C\exp(-cm)$, uniformly for all $\bx\in\calX$ we have:
    \begin{itemize}
        \item (PGD) The iterates $\{\bx_\ell\}_{\ell\ge 0}$ obtained by running (\ref{pgd11}) with $\bx_0\in\calK$ satisfy 
        \[\|\bx_\ell-\mu\bx\|_2\le \frac{\|\bx_0-\mu\bx\|_2}{2^\ell} + C\mathsf{B}_1(\phi,r,t,\zeta) +3\phi,\quad \forall \ell\ge 0.\]
        \item (Generalized Lasso) $\hat{\bx}_{\rm GL}$ in (\ref{Glasso11}) satisfies 
        \[\|\hat{\bx}_{\rm GL}-\mu\bx\|_2\le C\mathsf{B}_1(\phi,r,t,\zeta)+\phi.\]
    \end{itemize}
\end{theorem} 
\begin{proof}
    We provide a proof sketch here while deferring the  details to Section \ref{sec:convexproof}. Lemma \ref{lem:raic2convergeconvex} shows that the deterministic condition (\ref{uraicconvex}) is enough to imply the linear convergence of PGD to some statistical error. Such a condition is referred to as the restricted approximate invertibility condition (RAIC), as it requires that the gradient $\bh_{\bx}(\bu)$ well approximates the ideal descent step at $\bu$ (i.e., $\bu-\mu\bx$) over $\bu\in \calK$ under certain dual norm. While RAIC has been used to analyze PGD in the literature (e.g., \cite{matsumoto2024binary,chen2024optimal,chen2025unified}), a key adaptation here is to work with the dual norm on $(\calK-\mu\bx)_{[\phi]}$. This is enabled by Lemma \ref{lem:projectconvex} adapted from \cite{plan2017high} and is crucial for capturing the descent cone structure. It remains to establish the RAIC uniformly for all $\bx\in\calX$, which, with high probability, is further reduced to controlling a multiplier process in Lemma \ref{lem:furtherreduce}. Interestingly, by adapting some argument in \cite{genzel2023unified}, we find that the uniform recovery error rate of the generalized Lasso is dictated by the same multiplier process (see Lemma \ref{lem:glassosame}), thus enabling a simultaneous analysis of PGD and generalized Lasso. Finally, we control this multiplier process in Lemma \ref{lem:coreconvex} via a covering argument. 
\end{proof}

\begin{remark}
     Since  Assumption \ref{assump:fi} with $\tau=0$ recovers the $\gamma$ Lipschitz links, the term $\tau\sqrt{\delta}\bigg(\frac{\omega_1(\frac{\calK-\mu\calX}{\phi})}{\sqrt{m}}+\sqrt{\delta\log\frac{e}{\delta}}\bigg)$  in the  error can be interpreted as the extra cost of achieving uniform recovery under discontinuous links. 
\end{remark}

In the following, we specialize Theorem \ref{thm:recoveryconvex} to Assumption \ref{assump:sparsecone} and derive an explicit uniform recovery error rate. 

\begin{corollary}[Uniform sparse recovery with perfectly tuned $\ell_1$ constraint]
\label{cor:sparseconvex}
Under Assumptions \ref{assump:tsim}, \ref{assump:sparsecone} and \ref{assump:fi}, write
\(
    A:=1+\frac{\mu+\gamma}{\sigma}+\frac{1}{\beta}.\)
There exist universal constants $C>c>0$ such that the following holds.
If
\begin{align}
    m\ge C k\log\frac{Aen}{k},
    \label{sparseconvexsample}
\end{align}
let
\(
    D_m:=k\log\frac{en}{k}+k\log\frac{Am}{k\log(en/k)},
\)
then with probability at least $1-C\exp(-ck\log(en/k))$, uniformly for all
$\bx\in\calX$ we have:
\begin{itemize}
    \item (PGD) The iterates $\{\bx_\ell\}_{\ell\ge 0}$ obtained by
    running (\ref{pgd11}) with $\bx_0\in\calK$ satisfy
    \[
        \|\bx_\ell-\mu\bx\|_2
        \le \frac{\|\bx_0-\mu\bx\|_2}{2^\ell}
        +C\bigg( \sigma\sqrt{\frac{D_m}{m}}
    +\tau\frac{D_m}{m}\sqrt{\log\frac{em}{D_m}}\bigg),
        \qquad \forall \ell\ge 0.
    \]
    \item (Generalized Lasso) $\hat{\bx}_{\rm GL}$ in
    (\ref{Glasso11}) satisfies
    \[
        \|\hat{\bx}_{\rm GL}-\mu\bx\|_2
        \lesssim  \sigma\sqrt{\frac{D_m}{m}}
    +\tau\frac{D_m}{m}\sqrt{\log\frac{em}{D_m}}.
    \]
\end{itemize} 
\end{corollary}

\begin{proof} 
Under Assumption \ref{assump:sparsecone}, 
for every $\phi>0$ and $r\in(0,1)$ we have 
\begin{align}
    \omega_1\bigg(\frac{\calK-\mu\calX}{\phi}\bigg)
    \lesssim \sqrt{k\log(en/k)},
    \quad 
    \omega_1\bigg(\frac{\calX-\calX}{r}\bigg)
    \lesssim \sqrt{k\log(en/k)},
    \quad
    \log\calN(\calX;r)
    \le k\log\frac{9n}{kr}.
    \label{sparseconvexgeometry}
\end{align}
The second bound follows because   $\frac{\calX-\calX}{r}\subset \Sigma^n_{2k}$ for any $r>0.$ In Theorem \ref{thm:recoveryconvex}, we choose
\[
    t=\sqrt{k\log(en/k)},\qquad
    \phi=\sigma\sqrt{\frac{k\log(en/k)}{m}},\qquad
    \zeta=\frac{\lceil k\log(en/k)\rceil}{m},\qquad
    r=\frac{1}{A}\bigg(\frac{k\log(en/k)}{m}\bigg)^{3/2}. 
\]
Under (\ref{sparseconvexsample}) with a sufficiently large $C$,
we have $\frac{k\log(en/k)}{m}\in(0,1)$, $r\in(0,1)$ and
$\zeta\in\{i/m:i\in[m]\}$. Moreover, since $k\log(en/k)\ge 1$,  
\[
    \frac{k\log(en/k)}{m}\le\zeta\le \frac{2k\log(en/k)}{m},
    \qquad
    t^2=k\log(en/k)\le\lceil k\log(en/k)\rceil
    \le\zeta m\log\frac{e}{\zeta}.
\]
By (\ref{sparseconvexgeometry}) and the choice of $r$,
\[
    \log\calN(\calX;r)+k\log\frac{en}{k}
    \le k\log\frac{9nA}{k(k\log(en/k)/m)^{3/2}}+k\log\frac{en}{k}
    \lesssim k\log\frac{en}{k} + k\log\frac{Am}{k\log(en/k)}
    =D_m.
\]
Also, $\omega_r(\calX-\calX)\lesssim r\sqrt{k\log(en/k)}$ and
$2/\beta+\sqrt{2/\pi}\lesssim A$. Hence,
\begin{align*}
    &\bigg[\frac{2}{\beta}+\sqrt{\frac{2}{\pi}}\bigg]
    \bigg(
        \frac{\omega_r(\calX-\calX)}{\sqrt{\zeta m}}
        +r\sqrt{\log\frac{e}{\zeta}}
    \bigg)  \lesssim
    Ar\bigg(1+\sqrt{\log\frac{em}{k\log(en/k)}}\bigg)
    \\& =\bigg(\frac{k\log(en/k)}{m}\bigg)^{3/2}\bigg(1+\sqrt{\log\frac{em}{k\log(en/k)}}\bigg)
    \lesssim \frac{k\log(en/k)}{m},
\end{align*}
where the last inequality uses $\log(e/q)\le 1/q$ for $q\in(0,1)$.
Consequently, the parameter $\delta$ in
Theorem \ref{thm:recoveryconvex} satisfies
\[
    \delta\lesssim \frac{k\log(en/k)}{m}+\frac{D_m}{m}
    \lesssim\frac{D_m}{m}.
\]
We next verify that $\delta$ is sufficiently small.
Under (\ref{sparseconvexsample}),
\[
    \frac{D_m}{m}
    =\frac{k\log\frac{en}{k}+k\log A}{m}
    +\frac{k\log(\frac{m}{k\log(en/k)})}{m}
    \le\frac{1+\log C}{C},
\]
where we use $\frac{m}{k\log\frac{Aen}{k}}\ge C$, and hence large enough $C$ implies sufficiently small $\delta$.
Thus, increasing the universal constant $C$ in
(\ref{sparseconvexsample}) ensures the required condition
$\delta\le c$. The sample size condition involving the two
Gaussian widths also follows from
(\ref{sparseconvexgeometry}) and $m\ge Ck\log\frac{en}{k}$.
All conditions of Theorem \ref{thm:recoveryconvex} are therefore
satisfied.

It remains to bound $\mathsf{B}_1$.
Since $\frac{\mu+\gamma}{A}\le\sigma$, we have
\(
    (\mu+\gamma)r\le\sigma (\frac{k\log(en/k)}{m})^{3/2}.\)
Using $\delta\lesssim D_m/m$, $ k\log(en/k)\le D_m$, and the monotonicity of
$u\sqrt{\log(e/u)}$ on $(0,1)$, we also obtain
\begin{align*}
    &\tau\sqrt{\delta}
    \bigg(
        \frac{\omega_1((\calK-\mu\calX)/\phi)}{\sqrt m}
        +\sqrt{\delta\log\frac{e}{\delta}}
    \bigg)\\
    &\lesssim
    \tau\bigg(
        \frac{\sqrt{k\log(en/k)D_m}}{m}
        +\frac{D_m}{m}\sqrt{\log\frac{em}{D_m}}
    \bigg)
    \lesssim
    \tau\frac{D_m}{m}\sqrt{\log\frac{em}{D_m}}.
\end{align*}
Combining these bounds yields
\[
    \mathsf{B}_1(\phi,r,t,\zeta)+\phi
    \lesssim
    \sigma\sqrt{\frac{D_m}{m}}
    +\sigma \bigg(\frac{k\log\frac{en}{k}}{m}\bigg)^{3/2}
    +\tau\frac{D_m}{m}\sqrt{\log\frac{em}{D_m}}
    \lesssim  \sigma\sqrt{\frac{D_m}{m}}
    +\tau\frac{D_m}{m}\sqrt{\log\frac{em}{D_m}}.
\]
The two recovery guarantees now follow from
Theorem \ref{thm:recoveryconvex}, with failure probability bounded
by $C\exp(-ck\log(en/k))$. 
The proof is complete.
\end{proof}
\begin{remark} 
\label{rem:cost}
Ignoring the logarithmic factors, the uniform error rate in Corollary \ref{cor:sparseconvex} is at the order of $\sigma \sqrt{k/m}+ \tau k/m$ and matches the nonuniform error rate $O(\sigma\sqrt{k\log(en/k)/m})$ by Theorem \ref{thm:pvnonu}  whenever $\tau\le \sigma\sqrt{m/k}$ (up to logarithmic factors). Therefore, our new uniform recovery guarantees narrow the gap between Theorems \ref{thm:pvnonu} and \ref{thm:gsdiscu}---as explained in Section \ref{sec:gapidentify}---to the loss of some logarithmic factors. As such, while Genzel and Stollenwerk \cite{genzel2023unified} concluded that uniform recovery rate may substantially worsen under discontinuous links, our result suggests a complementary takeaway message: under a large class of discontinuous links, pursuing uniform recovery costs at most some logarithmic factors in the sample complexity and error rate. 
\end{remark} 

Nonetheless, we   view our Theorem \ref{thm:recoveryconvex} and  \cite{genzel2023unified} as complementary, as Theorems \ref{thm:gsconu}--\ref{thm:gsdiscu} may also yield tighter results in certain regimes. For instance, consider Assumption \ref{assump:sparsecone} with $\gamma$-Lipschitz links, Theorem \ref{thm:gsconu} implies the uniform error rate $O((\sigma+\mu+\gamma)\sqrt{k\log(en/k)/m})$. In contrast, due to a tradeoff in $r$ (which is a parameter arising from the covering argument in the proof), our Theorem \ref{thm:recoveryconvex} only yields uniform error rate with additional logarithmic factors. As such, we emphasize that the main contribution  of  Theorem \ref{thm:recoveryconvex} is to improve on \cite{genzel2023unified} under discontinuous links and a descent cone structure, with Assumption \ref{assump:sparsecone} being a canonical setting. 

\subsection{PGD with Cone $\calK$}
On the other hand, the next result is unique to  this paper and does not have a counterpart for generalized Lasso. 

\begin{theorem}[PGD with cone $\calK$] \label{thm:recoverycone}
Under Assumptions \ref{assump:tsim} and \ref{assump:fi}, assume that $\calK$ is a cone and $\calX\subset \calK$, there exist universal constants $C>c>0$ such that the following holds. For any $(r,t,\zeta)\in (0,1)\times[1,\infty)\times \{\frac{i}{m}:i\in[m]\}$ obeying $t^2\le \zeta m\log(e/\zeta)$ and \[\delta:=\zeta+ \frac{\log\calN(\calX;r)+t^2}{m}+r\bigg[\frac{2}{\beta}+\sqrt{\frac{2}{\pi}}\bigg] \bigg(\frac{\omega_1(\calK-\calK)}{\sqrt{\zeta m}}+\sqrt{\log\frac{e}{\zeta}}\bigg)\le c,\]
let
\[\mathsf{B}_2(r,t,\zeta):=\frac{\sigma\big(\omega_1( \calK-\calK)+\sqrt{\log\calN(\calX;r)+t^2}\big)}{\sqrt m}+ (\mu+\gamma )r + \tau\sqrt{\delta}
    \bigg(\frac{\omega_1(\calK-\calK)}{\sqrt m}
        +\sqrt{\delta\log\frac{e}{\delta}}\bigg),\]
        if 
        \( m\ge C\omega_1^2(\calK-\calK),\)
        then with probability at least $1-C\exp(-ct^2)-C\exp(-cm)$, uniformly for all $\bx\in\calX$, the iterates $\{\bx_\ell\}_{\ell\ge 0}$ obtained by running (\ref{pgd11}) with $\bx_0\in\calK$ satisfy 
        \[\|\bx_\ell-\mu\bx\|_2\le \frac{\|\bx_0-\mu\bx\|_2}{2^\ell} + C\mathsf{B}_2(r,t,\zeta),\qquad \forall \ell\ge 0.\]
\end{theorem}
\begin{proof}
    The proof is parallel to that of Theorem \ref{thm:recoveryconvex} up to some adaptation. To start, Lemma \ref{lem:conepgdraic} reduces the analysis to a deterministic RAIC in (\ref{uraiccone}).  Then Lemma \ref{lem:coneRAIC-reduce} further reduces the RAIC to controlling a multiplier process. Finally, this is achieved in Lemma \ref{lem:corecone} via a covering argument. The complete proof is provided in Section \ref{sec:proveconepgd}.   
\end{proof}

    We shall specialize    Theorem \ref{thm:recoverycone} to the setting in Assumption \ref{assump:iht} and derive an explicit uniform recovery error rate identical to Corollary \ref{cor:sparseconvex}. As such, through iterative hard thresholding, we again improve on \cite{genzel2023unified} in the uniform recovery of $k$-sparse signals under discontinuous links.

\begin{corollary}[Uniform sparse recovery via iterative hard thresholding]
\label{cor:sparsecone}
Under Assumptions \ref{assump:tsim}, \ref{assump:iht} and
\ref{assump:fi}, write
\(
    A:=1+\frac{\mu+\gamma}{\sigma}+\frac{1}{\beta}.\)
There exist universal constants $C>c>0$ such that the following holds.
If
\begin{align}
    m\ge Ck\log\frac{enA}{k},
    \label{sparseconesample}
\end{align}
let
\(
    D_m:=k\log\frac{en}{k}
    +k\log\frac{Am}{k\log(en/k)}.\)
Then with probability at least
$1-C\exp(-ck\log(en/k))$, the following holds uniformly for all
$\bx\in\Sigma^{n,*}_k$.
The iterates $\{\bx_\ell\}_{\ell\ge 0}$ obtained by running
\[
    \bx_{\ell+1}
    =H_k\bigg(\bx_\ell-\frac{1}{m}
    \bA^\top(\bA\bx_\ell-\by)\bigg),
    \qquad \ell\ge 0,
\]
with $\bx_0\in\Sigma^n_k$ satisfy
\[
    \|\bx_\ell-\mu\bx\|_2
    \le \frac{\|\bx_0-\mu\bx\|_2}{2^\ell}
    +C \bigg(\sigma\sqrt{\frac{D_m}{m}}
    +\tau\frac{D_m}{m}\sqrt{\log\frac{em}{D_m}}\bigg),
    \qquad \forall \ell\ge 0.
\] 
\end{corollary}
\begin{proof} 
Under Assumption \ref{assump:iht}, $\calK=\Sigma^n_k$ is a closed
cone containing $\calX=\Sigma^{n,*}_k$, and
$P_{\calK}=H_k$.
Moreover, $\calK-\calK=\Sigma^n_{\min\{2k,n\}}$, which gives
\(
    \omega_1(\calK-\calK)
    \lesssim\sqrt{k\log\frac{en}{k}}.\)
For every $r\in(0,1)$, we also have
\(
    \log\calN(\calX;r)\le k\log\frac{9n}{kr}.
\)
In Theorem \ref{thm:recoverycone}, we choose
\[
t^2=k\log\frac{en}{k},
    \qquad
    \zeta=\frac{\lceil k\log(en/k)\rceil}{m},
    \qquad
    r=\frac{1}{A}
    \bigg(\frac{k\log(en/k)}{m}\bigg)^{3/2}.\]
Under (\ref{sparseconesample}) with a sufficiently large $C$,
these parameters satisfy $t\ge 1$, $r\in(0,1)$ and
$\zeta\in\{i/m:i\in[m]\}$. Furthermore, in light of 
\(
    \frac{k\log(en/k)}{m}
    \le\zeta\le\frac{2k\log(en/k)}{m},\) \(
    t^2\le\zeta m\le\zeta m\log\frac{e}{\zeta},\)
we have \begin{align*}
    \log\calN(\calX;r)+t^2
    &\le
    k\log\bigg[
        \frac{9nA}{k}
        \bigg(\frac{m}{k\log(en/k)}\bigg)^{3/2}
    \bigg]
    +k\log\frac{en}{k}\lesssim D_m.
\end{align*}
Since $2/\beta+\sqrt{2/\pi}\lesssim A$, we obtain
\begin{align*}
    r\bigg[\frac{2}{\beta}+\sqrt{\frac{2}{\pi}}\bigg]
    \bigg(
        \frac{\omega_1(\calK-\calK)}{\sqrt{\zeta m}}
        +\sqrt{\log\frac{e}{\zeta}}
    \bigg) \lesssim
    Ar\bigg(
        1+\sqrt{\log\frac{em}{k\log(en/k)}}
    \bigg)
    \lesssim\frac{k\log(en/k)}{m},
\end{align*}
where the last inequality follows from the choice of $r$ and
$\log(eu)\le u$ for $u\ge 1$.
Consequently, the parameter $\delta$ in
Theorem \ref{thm:recoverycone} obeys
\(
    \delta\lesssim\frac{D_m}{m}.\)
We next verify that $\delta$ is sufficiently small.
By (\ref{sparseconesample}),
\[
    \frac{D_m}{m}
    =\frac{k\log(enA/k)}{m}
    +\frac{k}{m}\log\frac{m}{k\log(en/k)}
    \le\frac{1+\log C}{C}.
\]
Here, we use $m/(k\log(Aen/k))\ge C$, $\log(en/k)\ge 1$,
and the fact that $\log u/u$ is decreasing for $u\ge e$.
Thus, increasing the universal constant $C$ in
(\ref{sparseconesample}) ensures the required condition
$\delta\le c$.
Also,
\(
    m\gtrsim k\log\frac{en}{k}
    \gtrsim\omega_1^2(\calK-\calK).\)
All conditions of Theorem \ref{thm:recoverycone} are therefore
satisfied.

It remains to bound $\mathsf{B}_2$.
Since $\frac{\mu+\gamma}{A}\le\sigma$ and
$D_m\ge k\log(en/k)$, we have
\[
    (\mu+\gamma)r
    \le\sigma\bigg(\frac{k\log(en/k)}{m}\bigg)^{3/2}
    \le\sigma\sqrt{\frac{D_m}{m}}.
\]
Using $\delta\lesssim D_m/m$ and the monotonicity of
$u\sqrt{\log(e/u)}$ on $(0,1)$, the jump term satisfies
\begin{align*}
    &\tau\sqrt{\delta}
    \bigg(
        \frac{\omega_1(\calK-\calK)}{\sqrt m}
        +\sqrt{\delta\log\frac{e}{\delta}}
    \bigg)\\
    &\lesssim
    \tau\bigg(
        \frac{\sqrt{D_m k\log(en/k)}}{m}
        +\frac{D_m}{m}\sqrt{\log\frac{em}{D_m}}
    \bigg)
    \lesssim
    \tau\frac{D_m}{m}\sqrt{\log\frac{em}{D_m}}.
\end{align*}
Combining these estimates yields
\[
    \mathsf{B}_2(r,t,\zeta)
    \lesssim
    \sigma\sqrt{\frac{D_m}{m}}
    +\tau\frac{D_m}{m}\sqrt{\log\frac{em}{D_m}}.
\]
The convergence guarantee now follows from
Theorem \ref{thm:recoverycone}. 
The proof is complete.
\end{proof}

\subsection{Refining the Logarithmic Factors Under the ``sign'' Link} 
The results in previous sections still exhibit additional logarithmic factors compared to the nonuniform error rate.  While logarithmic factors are generally considered as less essential in the literature, a natural question is whether there exists a \emph{logarithmic} ``nonuniform-to-uniform recovery gap'' under discontinuous links. The following result refines Corollary \ref{cor:sparsecone} under
a canonical discontinuous sign link, that is $f(x)=\sign(x)=\mathbbm{1}(x\ge 0)-\mathbbm{1}(x<0)$ and establishes uniform recovery error rate $O(\sqrt{k\log(en/k)/m})$, which matches the nonuniform error rate up to a universal constant. 

\begin{theorem} \label{thm:sharp1b}
   Assume that $\bA\in \mathbb{R}^{m\times n}$ has i.i.d. $N(0,1)$ entries, $\by=\sign(\bA\bx)$ for $\bx\in\Sigma^{n,*}_k$. Suppose that we obtain $\{\bx_t\}_{t\ge 0}$ by running $\bx_{t+1}= H_k(\bx_t- \frac{1}{m}\bA^\top(\bA\bx_t-\by)),~t\ge 0$ with some $\bx_0\in \Sigma^{n}_k$. There exist universal constants $C>c>0$ such that, if $m\ge Ck\log\frac{en}{k}$, then with probability at least $1-C\exp(-ck\log\frac{en}{k})$, the following holds uniformly for all $\bx\in \Sigma^{n,*}_k$:
    \begin{align*}
        \|\bx_t-\sqrt{2/\pi}\bx\|_2 \le 
        \frac{\|\bx_0 - \sqrt{2/\pi}\bx\|_2}{2^t}+ \sqrt{\frac{Ck\log(en/k)}{m}},\qquad \forall t \ge 0.
    \end{align*}
\end{theorem}

\begin{proof}
We provide a technical overview here but relegate the complete argument to Appendix \ref{app:1bcs1prove}. By Lemma \ref{lem:coneRAIC-reduce} in the proof of Theorem \ref{thm:recoverycone}, it is enough to establish
\[\sup_{\bx\in\Sigma^{n,*}_k}\sup_{\bv\in \Sigma^{n,*}_{2k}}\bigg|\frac{1}{m}\sum_{i=1}^m\big(\!\sign(\ba_i^\top\bx)-\sqrt{2/\pi}\ba_i^\top\bx\big)\ba_i^\top\bv\bigg|\lesssim\sqrt{\frac{k\log(en/k)}{m}}.\]
In fact, Lemma \ref{lem:corecone} yields the bound $\tilde{O}(\sqrt{k/m})$ but involves additional logarithmic factors compared to the desired $O(\sqrt{k\log(en/k)/m})$, and here we have to use a finer argument. By using Lemma \ref{menproduct} to achieve the uniform concentration of $\frac{1}{m}\sum_{i=1}^m \bx^\top\ba_i\ba_i^\top\bv$ about the mean $\bx^\top\bv$, it remains to establish 
\[
\sup_{\bx\in\Sigma^{n,*}_k}\sup_{\bv\in \Sigma^{n,*}_{2k}}\bigg|\frac{1}{m}\sum_{i=1}^m \sign(\ba_i^\top \bx)\ba_i^\top \bv-\sqrt{2/\pi}\bx^\top \bv\bigg|\lesssim\sqrt{\frac{k\log(en/k)}{m}}\,.
\]
Our key insight is to work conditionally on the sign functions and find a suitable metric that is essentially an interpolation between the standard $\ell_2$ metric induced by the linear component $\ba_i^\top \bv$ and the Lipschitz metric with respect to the empirical measure given by $\sum_{i=1}^m \sign(\ba_i^\top \bx)-\sign(\ba_i^\top \bx')$; importantly, covering number estimates are available through invoking VC theory. Finally, once the covering number under this interpolation metric is controlled, we can apply Dudley's inequality to control the supremum of the process of interest. 
\end{proof}

\begin{remark}
    We shall emphasize that the goal of Theorem \ref{thm:sharp1b} is not merely to show that a sharper uniform error rate holds for the very specific sign link---indeed, this specific problem is known as one-bit compressed sensing, for which the sharper $\tilde{O}(k/m)$ uniform error rate can be attained \cite{jacques2013robust,matsumoto2024binary,chen2024optimal}. Rather, the main goal of Theorem \ref{thm:sharp1b} is to shed  light on the understanding of the additional logarithmic factors in our previous uniform recovery guarantees---we believe that they are largely technical artifacts and can indeed be removed in a much broader regime. We leave such refinement under more general discontinuous links or in the setting of Corollary \ref{cor:sparseconvex} for future work. 
\end{remark}

\section{Proof of Theorem \ref{thm:recoveryconvex}} \label{sec:convexproof}
For $\calU\subset \mathbb{R}^n$ and $\phi>0$, let $\calU_{[\phi]}:=\calU\cap \phi\mathbb{B}_2^n$. Following \cite{friedlander2021nbiht}, we also let the \emph{dual norm} of $\bw\in \mathbb{R}^n$ with respect to $\calU$ be
\[\|\bw\|_{\calU^\circ}:=\sup_{\bu\in \calU}\,\langle \bw,\bu\rangle.\]
Strictly, this is not a norm, while we note the following useful properties:
\begin{gather}\label{dualposilinear} \|\lambda\bw\|_{\calU^\circ} = \lambda \|\bw\|_{\calU^\circ},\quad \forall \lambda\ge 0;\\\label{duaitri}
    \|\bw_1+\bw_2\|_{\calU^\circ} \le \|\bw_1\|_{\calU^\circ}  + \|\bw_2\|_{\calU^\circ}; \\ \label{dualcontain}
    \|\bw\|_{\calU^\circ}\le \|\bw\|_{\calU_1^\circ},\quad \forall \calU\subset\calU_1.
\end{gather}

We say a set $\calU$ is \emph{star-shaped}, if $\bu\in\calU$ implies $\lambda \bu \in \calU$ for any $\lambda\in[0,1]$. For instance, a cone and a convex set containing $0$ are star-shaped.
\begin{lemma}\label{lem:projectconvex}
    Let $\calK$ be a closed convex set, $\bz\in\calK$, and $\bw\in \mathbb{R}^n$. Then for every $\phi>0$, 
    \[\|P_{\calK}(\bw)-\bz\|_2\le \max \bigg\{\phi,\frac{2}{\phi}\|\bw-\bz\|_{(\calK-\bz)^\circ_{[\phi]}}\bigg\}.\]
\end{lemma}
\begin{proof}
    The proof is adapted from \cite[Corollary 8.3]{plan2017high}. By the definition of $P_{\calK}$, we have
    \(\|P_{\calK}(\bw)-\bw\|_2^2\le \|\bw-\bz\|_2^2\). An expansion 
      gives
      \[\|P_{\calK}(\bw)-\bz\|_2^2\le 2\langle P_{\calK}(\bw)-\bz,\bw-\bz\rangle.\]
      Given any $\phi>0$. If $\|P_{\calK}(\bw)-\bz\|_2\le \phi$, then the conclusion is trivial. Hence we assume $\|P_{\calK}(\bw)-\bz\|_2\ge \phi$, and then we can write 
      \begin{align}
          \|P_{\calK}(\bw)-\bz\|_2 \le 2\bigg\langle \frac{P_{\calK}(\bw)-\bz}{\|P_{\calK}(\bw)-\bz\|_2},\bw-\bz\bigg\rangle.\label{234}
      \end{align}
       Since  $\bz\in\calK$ where $\calK$ is convex, we have that $\calK-\bz$ is a convex set containing $0$, and is therefore star-shaped. Combining with $\|P_{\calK}(\bw)-\bz\|_2\ge \phi$, we obtain 
      \[\frac{\calK-\bz}{\|P_{\calK}(\bw)-\bz\|_2}\subset\frac{\calK-\bz}{\phi}.\]
     By  $P_{\calK}(\bw)\in\calK$, we have $\frac{P_{\calK}(\bw)-\bz}{\|P_{\calK}(\bw)-\bz\|_2} \in \frac{\calK-\bz}{\|P_{\calK}(\bw)-\bz\|_2}\subset\frac{\calK-\bz}{\phi}$. Further combining with $\frac{P_{\calK}(\bw)-\bz}{\|P_{\calK}(\bw)-\bz\|_2} \in \mathbb{S}^{n-1}$, we obtain \[\frac{P_{\calK}(\bw)-\bz}{\|P_{\calK}(\bw)-\bz\|_2} \in \frac{\calK-\bz}{\phi}\cap \mathbb{S}^{n-1}\subset\frac{(\calK-\bz)_{[\phi]}}{\phi}.\] We then continue from (\ref{234}) to obtain 
     \[ \|P_{\calK}(\bw)-\bz\|_2 \le \frac{2}{
     \phi
     }\sup_{\bu\in (\calK-\bz)_{[\phi]}}\langle\bw-\bz,\bu\rangle = \frac{2}{\phi}\|\bw-\bz\|_{(\calK-\bz)_{[\phi]}^\circ}.\]
Hence, the conclusion also holds when $\|P_{\calK}(\bw)-\bz\|_{2}\ge \phi$. The proof is complete. 
\end{proof}
 
\begin{lemma} \label{lem:raic2convergeconvex}
Suppose that $\calK$ is convex, $\mu\calX\subset \calK$, and $\{\bx_t\}_{t\ge 0}$ is obtained by running $\bx_{t+1}=P_{\calK}(\bx_t-\bh_{\bx}(\bx_t)),~t\ge 0$
with $\bx_0\in \calK$. 
    For any $\phi>0$, if 
    \begin{align}\label{uraicconvex}
        \frac{1}{\phi}\big\|\bu-\mu\bx- \bh_{\bx}(\bu)\big\|_{(\calK-\mu\bx)^\circ_{[\phi]}}\le \frac{1}{4}\|\bu-\mu\bx\|_2 + B,\quad \forall \bu\in\calK,~\bx\in \calX,
    \end{align}
    then the following holds uniformly for all $\bx\in\calX$:
    \[\|\bx_t-\mu\bx\|_2\le \frac{\|\bx_0-\mu\bx\|_2}{2^t} + 4B+2\phi,\quad \forall t\ge 0.\]
\end{lemma}
\begin{proof}
      By definition we have $\bx_t\in \calK$ for any $t\ge 0$. By using Lemma \ref{lem:projectconvex} and (\ref{uraicconvex}), it holds for any $t\ge 0$ that 
      \begin{align*}
          \|\bx_{t+1}-\mu\bx\|_2 &= \big\|P_{\calK}(\bx_t-\bh_{\bx}(\bx_t))-\mu\bx\big\|_2 \\
          &\le \phi+\frac{2}{\phi}\big\|\bx_t-\mu\bx-\bh_{\bx}(\bx_t)\big\|_{(\calK-\mu\bx)^\circ_{[\phi]}}\\&\le \frac{1}{2}\|\bx_t-\mu\bx\|_2+ 2B+\phi.
      \end{align*}
      Solving the recurrence inequality yields the claim. 
\end{proof}

\begin{lemma}[Reduction of PGD analysis to bounding an empirical process]  \label{lem:furtherreduce}
   Let $\phi>0$. If $m\gtrsim \omega^2_{1}(\frac{\calK-\mu \calX}{\phi})$, then with probability at least $1-C\exp(-cm)$, it holds for all $(\bu,\bx)\in\calK\times\calX$ that \begin{align}\label{raicforconvexpgd}
         \frac{1}{\phi}\big\|\bu-\mu\bx- \bh_{\bx}(\bu)\big\|_{(\calK-\mu\bx)^\circ_{[\phi]}}\le \frac{1}{4}\|\bu-\mu\bx\|_2 + \frac{\phi}{4} + \sup_{\bx\in\calX}\frac{\|\bA^\top(\bm{f}(\bA\bx)-\mu\bA\bx)\|_{(\calK-\mu\calX)^\circ_{[\phi]}}}{m\phi}. 
    \end{align}
    Consequently, in light of Lemma \ref{lem:raic2convergeconvex},   uniformly for all $\bx\in\calX$, $\{\bx_t\}_{t\ge 0}$ in Lemma \ref{lem:raic2convergeconvex} satisfies 
    \[\|\bx_t-\mu\bx\|_2\le \frac{\|\bx_0-\mu\bx\|_2}{2^t} + \sup_{\bx\in\calX}\sup_{\bv\in (\frac{\calK-\mu\calX}{\phi})_{[1]}}\frac{4}{m}\sum_{i=1}^m\big(f_i(\ba_i^\top\bx)-\mu\ba_i^\top\bx\big)\ba_i^\top\bv+3\phi,\quad \forall t\ge 0.\]
\end{lemma}
\begin{proof}
    By (\ref{hxshort}), (\ref{dualposilinear}), (\ref{duaitri}), (\ref{dualcontain}) and $(\calK-\mu\bx)_{[\phi]}\subset (\calK-\mu\calX)_{[\phi]}$, it holds for all $(\bu,\bx)\in\calK\times\calX$ that  
    \begin{align*}
        \frac{1}{\phi}\big\|\bu-\mu\bx- \bh_{\bx}(\bu)\big\|_{(\calK-\mu\bx)^\circ_{[\phi]}} &\le \underbrace{\frac{1}{\phi}\bigg\|\bu-\mu\bx - \frac{1}{m}\bA^\top\bA(\bu-\mu\bx)\bigg\|_{(\calK-\mu\calX)^\circ_{[\phi]}}}_{:=\Xi_1} \\&+  \sup_{\bx\in\calX}\frac{\|\bA^\top(\bm{f}(\bA\bx)-\mu\bA\bx)\|_{(\calK-\mu\calX)^\circ_{[\phi]}}}{m\phi}.
    \end{align*}
    Moreover, for any $(\bu,\bx)\in\calK\times\calX$, 
    \begin{align*}
        &\Xi_1= \frac{1}{\phi}\sup_{\bw \in (\calK-\mu\calX)_{[\phi]}} \bigg\langle \bw,\bu-\mu\bx - \frac{1}{m}\bA^\top\bA(\bu-\mu\bx)\bigg\rangle\\
        &\le \frac{1}{\phi}\sup_{\bw \in (\calK-\mu\calX)_{[\phi]}} \bigg\langle \bw,\bu-\mu\bx - \frac{1}{m}\bA^\top\bA(\bu-\mu\bx)\bigg\rangle \mathbbm{1}(\|\bu-\mu\bx\|_2\le \phi)  \\
        & + \frac{1}{\phi}\sup_{\bw \in (\calK-\mu\calX)_{[\phi]}} \bigg\langle \bw,\bu-\mu\bx - \frac{1}{m}\bA^\top\bA(\bu-\mu\bx)\bigg\rangle \mathbbm{1}(\|\bu-\mu\bx\|_2>\phi)\\:&= \Xi_{1,1} + \Xi_{1,2}. 
    \end{align*}
    If $\|\bu-\mu\bx\|_2\le \phi$, then $\bu-\mu\bx \in (\calK-\mu\calX)_{[\phi]}$. Therefore letting $\bv=\bu-\mu\bx$ gives
    \begin{align*}
        \Xi_{1,1}\le \frac{1}{\phi}\sup_{\substack{\bw \in (\calK-\mu\calX)_{[\phi]}\\\bv \in (\calK-\mu\calX)_{[\phi]}}} \bigg\langle \bw, \bv - \frac{1}{m}\bA^\top\bA\bv \bigg\rangle.
    \end{align*}
    If $\|\bu-\mu\bx\|_2 >\phi$, then since $\calK-\mu\bx$ is a convex set containing $0$ and therefore star-shaped, \[\frac{\bu-\mu\bx}{\|\bu-\mu\bx\|_2}\in \frac{\calK -\mu\bx}{\|\bu-\mu\bx\|_2} \cap \mathbb{B}_2^n\subset \frac{\calK-\mu\bx}{\phi}\cap \mathbb{B}_2^n\subset  \frac{\calK-\mu\calX}{\phi}\cap \mathbb{B}_2^n = \frac{(\calK-\mu\calX)_{[\phi]}}{\phi}.\]
    Therefore, letting $\bv=\frac{\phi(\bu-\mu\bx)}{\|\bu-\mu\bx\|_2}\in (\calK-\mu\calX)_{[\phi]}$ gives
    \[\Xi_{1,2}\le \frac{\|\bu-\mu\bx\|_2}{\phi^2}\sup_{\substack{\bw \in (\calK-\mu\calX)_{[\phi]}\\\bv \in (\calK-\mu\calX)_{[\phi]}}} \bigg\langle \bw, \bv - \frac{1}{m}\bA^\top\bA\bv \bigg\rangle,\quad \forall (\bu,\bx)\in\calK\times \calX.\]
    Combining these pieces, it holds for all $ (\bu,\bx)\in\calK\times \calX$ that 
    \begin{align}\nonumber
        \Xi_1 &\le \bigg(\frac{1}{\phi}+\frac{\|\bu-\mu\bx\|_2}{\phi^2}\bigg)\sup_{\substack{\bw \in (\calK-\mu\calX)_{[\phi]}\\\bv \in (\calK-\mu\calX)_{[\phi]}}} \bigg\langle \bw, \bv - \frac{1}{m}\bA^\top\bA\bv \bigg\rangle \\
        &\le \bigg(\frac{1}{\phi}+\frac{\|\bu-\mu\bx\|_2}{\phi^2}\bigg)\underbrace{\sup_{\substack{\bw \in (\calK-\mu\calX)_{[\phi]}\\\bv \in (\calK-\mu\calX)_{[\phi]}}}\bigg|\frac{1}{m}\sum_{i=1}^m\big(\bw^\top\ba_i\ba_i^\top\bv - \mathbb{E}[\bw^\top\ba_i\ba_i^\top\bv]\big)\bigg|}_{:=\Xi_{1,3}}. \label{9763}
    \end{align}
    By Lemma \ref{menproduct}, for any $t\ge 1$ we have
    \[
        \mathbb{P}\bigg(\Xi_{1,3}\lesssim \frac{[\omega_\phi(\calK-\mu\calX)+t\phi]^2}{m}+\frac{\phi[\omega_\phi(\calK-\mu\calX)+t\phi]}{\sqrt{m}}\bigg)\ge 1-2\exp(-ct^2).
    \]
    Under $m\gtrsim \frac{\omega^2_\phi(\calK-\mu\calX)}{\phi^2}=\omega_1^2(\frac{\calK-\mu\calX}{\phi})$, we set $t=c\sqrt{m}$ to obtain that $\Xi_{1,3}\le\frac{\phi^2}{4}$ with the promised probability. Substituting this into (\ref{9763}) yields $\Xi_1\le \frac{1}{4}\|\bu-\mu\bx\|_2 + \frac{\phi}{4}$. (\ref{raicforconvexpgd}) follows. The convergence guarantee is then a consequence of Lemma \ref{lem:raic2convergeconvex} and 
   \[\sup_{\bx\in\calX}\frac{\|\bA^\top(\bm{f}(\bA\bx)-\mu\bA\bx)\|_{(\calK-\mu\calX)^\circ_{[\phi]}}}{m\phi}=\sup_{\bx\in\calX}\sup_{\bv\in (\frac{\calK-\mu\calX}{\phi})_{[1]}}\frac{1}{m}\sum_{i=1}^m\big(f_i(\ba_i^\top\bx)- \mu\ba_i^\top\bx\big)\ba_i^\top\bv.\]
   The proof is complete. 
\end{proof}
\begin{lemma}[Reduction of generalized Lasso analysis to bounding an empirical process] \label{lem:glassosame} Let $\phi>0$. If $m\gtrsim \omega^2_1(\frac{\calK-\mu\calX}{\phi}) $, then with probability at least $1-C\exp(-cm)$, it holds for all $\bx\in \calX$ that 
\begin{align}
    \|\hat{\bx}_{\rm GL}-\mu\bx\|_2 \le \max\bigg\{\phi,\, \sup_{\bx\in\calX}\sup_{\bv\in (\frac{\calK-\mu\calX}{\phi})_{[1]}}\frac{4}{m}\sum_{i=1}^m\big(f_i(\ba_i^\top\bx)-\mu\ba_i^\top\bx \big)\ba_i^\top\bv\bigg\},\label{glassobound}
\end{align}
where $\hat{\bx}_{\rm GL}$ is the output of the generalized Lasso in (\ref{Glasso11}). 
\end{lemma}
\begin{proof}  
     For the given $\phi>0$. Under \(
          m\gtrsim \omega^2_1\big(\frac{\calK-\mu\calX}{\phi}\big)=\frac{\omega_{\phi}^2(\calK-\mu\calX)}{\phi^2}, 
      \) by the argument in Step 1 of \cite[Section 7]{genzel2023unified}, it holds with probability at least $1-C\exp(-cm)$ that 
      \begin{align}
          \inf_{\bv\in \frac{\calK-\mu\calX}{\phi}\cap \mathbb{S}^{n-1}}\frac{1}{m}\|\bA\bv\|_2^2 \ge \frac{1}{2}.\label{quadraticlower}
      \end{align}
      We proceed on this event, and we only treat $\bx\in\calX$ such that $\|\hat{\bx}_{\rm GL} - \mu\bx\|_2 \ge \phi$ without loss of generality (for the remaining  $\bx$, we have $\|\hat{\bx}_{\rm GL}-\mu\bx\|_2\le \phi$, and we are done). By $\hat{\bx}_{\rm GL}\in \calK$ and $\mu\bx\in\calX$,  
      \[\frac{\hat{\bx}_{\rm GL}-\mu\bx}{\|\hat{\bx}_{\rm GL}-\mu\bx\|_2}\in  \frac{\calK-\mu\calX}{\|\hat{\bx}_{\rm GL}-\mu\bx\|_2}\cap \mathbb{S}^{n-1}\subset\frac{\calK-\mu\calX}{\phi}\cap \mathbb{S}^{n-1},\]
      where the inclusion holds because $\calK-\mu\bz$ is a convex set containing $0$, and hence star-shaped, and therefore 
      \begin{align}
          \frac{\calK-\mu\calX}{\|\hat{\bx}_{\rm GL}-\mu\bx\|_2}= \bigcup_{\bz\in\calX} \frac{\calK-\mu\bz}{\|\hat{\bx}_{\rm GL}-\mu\bx\|_2}\subset  \bigcup_{\bz\in\calX} \frac{\calK-\mu\bz}{\phi} = \frac{\calK-\mu\calX}{\phi}.\label{rangeof1}      
    \end{align}      
    By the optimality of $\hat{\bx}_{\rm GL}$, we have
      \[\|\bA\hat{\bx}_{\rm GL}-\by\|_2 ^2\le \|\mu\bA\bx-\by\|_2^2.\]
      Using $\by=\bm{f}(\bA\bx)$, simple algebra finds 
      \[\|\bA(\hat{\bx}_{\rm GL}-\mu\bx)\|_2^2 \le 2\langle \bA(\hat{\bx}_{\rm GL}-\mu\bx),\bm{f}(\bA\bx)-\mu\bA\bx\rangle,\]
      that is 
      \[\|\hat{\bx}_{\rm GL}-\mu\bx\|_2\left\|\bA\frac{\hat{\bx}_{\rm GL}-\mu\bx}{\|\hat{\bx}_{\rm GL}-\mu\bx\|_2}\right\|_2^2\le 2\left\langle \bA\frac{\hat{\bx}_{\rm GL}-\mu\bx}{\|\hat{\bx}_{\rm GL}-\mu\bx\|_2},\bm{f}(\bA\bx)-\mu\bA\bx\right\rangle.\]
      By (\ref{rangeof1}), this gives
      \[\|\hat{\bx}_{\rm GL}-\mu\bx\|_2 \inf_{\bv\in \frac{\calK-\mu\calX}{\phi}\cap \mathbb{S}^{n-1}} \|\bA\bv\|_2^2\le 2\sup_{\bv\in (\frac{\calK-\mu\calX}{\phi})_{[1]}}\langle \bA\bv,\bm{f}(\bA\bx)-\mu\bA\bx\rangle.\]
      Further using (\ref{quadraticlower}) and rearranging, 
      \[\|\hat{\bx}_{\rm GL}-\mu\bx\|_2\le \frac{4}{m}\sup_{\bv\in (\frac{\calK-\mu\calX}{\phi})_{[1]}}\langle \bA\bv,\bm{f}(\bA\bx)-\mu\bA\bx\rangle.\]
      Since every step holds uniformly for all $\bx$ such that $\|\hat{\bx}_{\rm GL}-\mu\bx\|_2\ge \phi$, we take supremum the obtain the following that holds uniformly for such $\bx$:
      \[\|\hat{\bx}_{\rm GL}-\mu\bx\|_2\le \frac{4}{m}\sup_{\bx\in\calX}\sup_{\bv\in (\frac{\calK-\mu\calX}{\phi})_{[1]}}\langle \bA\bv,\bm{f}(\bA\bx)-\mu\bA\bx\rangle.\]
      This shows that $\|\hat{\bx}_{\rm GL}-\mu\bx\|_2$ is uniformly bounded by the second branch in (\ref{glassobound}). The proof is complete. 
\end{proof}
\begin{lemma}\label{lem:coreconvex} Given $\phi>0.$
    Under Assumptions \ref{assump:tsim} and \ref{assump:fi}, there exist universal constants $C>c>0$ such that the following holds. For any $(r,t,\zeta)\in (0,1)\times [1,\infty)\times \{\frac{i}{m}:i\in[m]\}$ obeying $t^2\le \zeta m \log(e/\zeta)$ and 
    \[ \delta:=\zeta+\frac{\log\calN(\calX;r)+t^2}{m}
    +\bigg[\frac{2}{\beta}+\sqrt{\frac{2}{\pi}}\bigg]\bigg(
        \frac{\omega_r(\calX-\calX)}{\sqrt{\zeta m}}
        +r\sqrt{\log\frac{e}{\zeta}}
    \bigg)
    \le c,\]
    if 
    \[m\ge C\bigg[ \omega_1^2\bigg(\frac{\calK-\mu\calX}{\phi}\bigg)+ \omega_1^2\bigg(\frac{\calX-\calX}{r}\bigg)\bigg],\]
    then with probability at least $1-C\exp(-ct^2)$,
    \begin{align*}
        &\sup_{\bx\in\calX}\sup_{\bv\in (\frac{\calK-\mu\calX}{\phi})_{[1]}}\frac{4}{m}\sum_{i=1}^m\big(f_i(\ba_i^\top\bx)-\mu\ba_i^\top\bx\big)\ba_i^\top\bv \\
        &\lesssim  \frac{\sigma\big(\omega_1(\frac{\calK-\mu\calX}{\phi})+\sqrt{\log\calN(\calX;r)+t^2}\big)}{\sqrt m}+ (\mu+\gamma )r + \tau\sqrt{\delta}
    \bigg(\frac{\omega_1(\frac{\calK-\mu\calX}{\phi})}{\sqrt m}
        +\sqrt{\delta\log\frac{e}{\delta}}\bigg). 
    \end{align*}
\end{lemma}
\begin{proof}
    This is a direct consequence of invoking Lemma \ref{lem:uniformmultiplier} with $\calV=(\frac{\calK-\mu\calX}{\phi})_{[1]}$. 
\end{proof}
\section{Proof of Theorem \ref{thm:recoverycone}}\label{sec:proveconepgd}
\begin{lemma}\label{lem:coneproj}
    Let $\calK$ be a closed cone, $\bz\in \calK$, and $\bw\in \mathbb{R}^n$. Then 
    \[\|P_{\calK}(\bw)-\bz\|_2\le 2\|\bw-\bz\|_{(\calK-\calK)^\circ_{[1]}}\]
\end{lemma}
\begin{proof}
    Since $\|\bw-\bz\|_{(\calK-\calK)^\circ_{[1]}}\ge 0$, the claim holds trivially if $P_{\calK}(\bw)=\bz$. Hence we assume $P_{\calK}(\bw)\ne\bz$. By revisiting the argument in Lemma \ref{lem:projectconvex}, we obtain (\ref{234}). The result follows by $\frac{P_{\calK}(\bw)-\bz}{\|P_{\calK}(\bw)-\bz\|_2}\in (\calK-\calK)_{[1]}$, which then gives 
    \[\bigg\langle \frac{P_{\calK}(\bw)-\bz}{\|P_{\calK}(\bw)-\bz\|_2},\bw-\bz\bigg\rangle\le \|\bw-\bz\|_{(\calK-\calK)^\circ_{[1]}}.\]
\end{proof}

\begin{lemma}\label{lem:conepgdraic}
Suppose that $\calK$ is a closed cone, $\calX\subset \calK$, and $\{\bx_t\}_{t\ge 0}$ is obtained by running $\bx_{t+1}=P_{\calK}(\bx_t-\bh_{\bx}(\bx_t)),~t\ge 0$ with $\bx_0\in\calK$. If 
 \begin{align}\label{uraiccone}
       \big\|\bu-\mu\bx- \bh_{\bx}(\bu)\big\|_{(\calK-\calK)^\circ_{[1]}}\le \frac{1}{4}\|\bu-\mu\bx\|_2 + B,\quad \forall \bu\in\calK,~\bx\in \calX,
    \end{align}
    then the following holds uniformly for all $\bx\in\calX$:
    \[\|\bx_t-\mu\bx\|_2\le \frac{\|\bx_0-\mu\bx\|_2}{2^t} + 4B,\quad \forall t\ge 0.\]
\end{lemma}
\begin{proof}
    By definition, $\bx_t\in\calK$ for any $t\ge 0.$ Using Lemma \ref{lem:coneproj} and (\ref{uraiccone}), 
    \begin{align*}
        &\|\bx_{t+1}-\mu\bx\|_2 = \big\|P_{\calK}(\bx_t-\bh_{\bx}(\bx_t))-\mu\bx\big\|_2\\
        & \le 2\big\|\bx_t-\mu\bx-\bh_{\bx}(\bx_t)\big\|_{(\calK-\calK)^\circ_{[1]}} \le 
        \frac{1}{2}\|\bx_t-\mu\bx\|_2 + 2B
    \end{align*}
    holds for any $t\ge 0.$ Solving the recurrence inequality yields the claim. 
\end{proof}
\begin{lemma} \label{lem:coneRAIC-reduce}
    If $m\gtrsim \omega_1^2(\calK-\calK)$, then with probability at least $1-C\exp(-cm)$, it holds for all $(\bu,\bx)\in \calK\times \calX$ that
    \begin{align}\label{proveraic-cone}
        \big\|\bu-\mu\bx- \bh_{\bx}(\bu)\big\|_{(\calK-\calK)_{[1]}^\circ}\le \frac{1}{4}\|\bu-\mu\bx\|_2+ \sup_{\bx\in\calX}\frac{\|\bA^\top(\bm{f}(\bA\bx)-\mu\bA\bx)\|_{(\calK-\calK)_{[1]}^\circ}}{m}. 
    \end{align} 
    Consequently, in light of Lemma \ref{lem:conepgdraic}, uniformly for all $\bx\in\calX$, $\{\bx_t\}_{t\ge 0}$ therein satisfies 
    \[\|\bx_t-\mu\bx\|_2\le \frac{\|\bx_0-\mu\bx\|_2}{2^t} + \sup_{\bx\in\calX}\sup_{\bv\in (\calK-\calK)_{[1]}}\frac{4}{m}\sum_{i=1}^m\big(f_i(\ba_i^\top\bx)-\mu\ba_i^\top\bx\big)\ba_i^\top\bv,\quad \forall t\ge 0.\]
\end{lemma}
\begin{proof}
    By  (\ref{hxshort}), (\ref{dualposilinear}), (\ref{duaitri}), it holds for all $(\bu,\bx)\in \calK\times\calX$ that
    \[\big\|\bu-\mu\bx- \bh_{\bx}(\bu)\big\|_{(\calK-\calK)_{[1]}^\circ}\le \bigg\|\bu-\mu\bx - \frac{1}{m}\bA^\top\bA(\bu-\mu\bx)\bigg\|_{(\calK-\calK)^\circ_{[1]}}+ \sup_{\bx\in\calX}\frac{\|\bA^\top(\bm{f}(\bA\bx)-\mu\bA\bx)\|_{(\calK-\calK)_{[1]}^\circ}}{m}.\]
    It remains to bound the first term. If $\bu= \mu\bx$, then this term vanishes. If $\bu\ne \mu\bx$, then by $\frac{\bu-\mu\bx}{\|\bu-\mu\bx\|_2}\in (\calK-\calK)_{[1]}$,   
    \[\bigg\|\bu-\mu\bx - \frac{1}{m}\bA^\top\bA(\bu-\mu\bx)\bigg\|_{(\calK-\calK)^\circ_{[1]}} \le \|\bu-\mu\bx\|_2 \sup_{\substack{\bw\in (\calK-\calK)_{[1]}\\\bv\in (\calK-\calK)_{[1]}}}\bigg|\frac{1}{m}\sum_{i=1}^m\bw^\top\ba_i\ba_i^\top\bv - \mathbb{E}[\bw^\top\ba_i\ba_i^\top\bv]\bigg|.\]
    By invoking Lemma \ref{menproduct} with $t=c_0\sqrt{m}$ for small enough $c_0$, under $m\gtrsim \omega^2_1(\calK-\calK)$, we establish that 
    \[\sup_{\substack{\bw\in (\calK-\calK)_{[1]}\\\bv\in (\calK-\calK)_{[1]}}}\bigg|\frac{1}{m}\sum_{i=1}^m\bw^\top\ba_i\ba_i^\top\bv - \mathbb{E}[\bw^\top\ba_i\ba_i^\top\bv]\bigg|\le \frac{1}{4}\]
    holds with promised probability. (\ref{proveraic-cone}) follows.  Notice that 
    \[\sup_{\bx\in\calX}\frac{\|\bA^\top(\bm{f}(\bA\bx)-\mu\bA\bx)\|_{(\calK-\calK)_{[1]}^\circ}}{m}=\sup_{\bx\in\calX}\sup_{\bv\in (\calK-\calK)_{[1]}}\frac{1}{m}\sum_{i=1}^m\big(f_i(\ba_i^\top\bx)-\mu\ba_i^\top\bx\big)\ba_i^\top\bv,\]
    the convergence guarantee follows from Lemma \ref{lem:conepgdraic}. 
\end{proof}

\begin{lemma}\label{lem:corecone}
    Under Assumptions \ref{assump:tsim} and \ref{assump:fi}, assume that $\calK$ is a cone containing $\calX$, then there exist universal constants $C>c>0$ such that the following holds. For any $(r,t,\zeta)\in (0,1)\times [1,\infty)\times\{\frac{i}{m}:i\in[m]\}$ obeying $t^2\le \zeta m \log(e/\zeta)$ and 
    \[\delta:=\zeta+ \frac{\log\calN(\calX;r)+t^2}{m}+r\bigg[\frac{2}{\beta}+\sqrt{\frac{2}{\pi}}\bigg] \bigg(\frac{\omega_1(\calK-\calK)}{\sqrt{\zeta m}}+\sqrt{\log\frac{e}{\zeta}}\bigg)\le c,\]
    if $m\ge C\omega^2_1(\calK-\calK)$, then with probability at least $1-C\exp(-ct^2),$ 
    \begin{align*}
        &\sup_{\bx\in\calX}\sup_{\bv\in (\calK-\calK)_{[1]}}\frac{1}{m}\sum_{i=1}^m\big(f_i(\ba_i^\top\bx)-\mu\ba_i^\top\bx\big)\ba_i^\top\bv\\
        &\lesssim \frac{\sigma\big(\omega_1( \calK-\calK)+\sqrt{\log\calN(\calX;r)+t^2}\big)}{\sqrt m}+ (\mu+\gamma )r + \tau\sqrt{\delta}
    \bigg(\frac{\omega_1(\calK-\calK)}{\sqrt m}
        +\sqrt{\delta\log\frac{e}{\delta}}\bigg).
    \end{align*}
\end{lemma}
\begin{proof}
    This directly follows from invoking Theorem \ref{lem:uniformmultiplier} with $\calV=(\calK-\calK)_{[1]}$ and using $(\calX-\calX)_{[r]}\subset r(\calK-\calK)_{[1]}$ (due to $\calX\subset \calK$).
\end{proof}

\section{Discussions}\label{sec:discuss}
Generalized Lasso is an accurate solver for Gaussian single index model with unknown links. However, under discontinuous links, there is an essential gap between its nonuniform recovery error rate and uniform recovery error rate. In particular, the nonuniform guarantee \cite{plan2016generalized} holds only for a fixed signal and scales as $O(m^{-1/2})$ under a descent cone structure, while the uniform guarantee \cite{genzel2023unified} decays with $m$ in a much slower rate $O(m^{-1/4})$. The authors of  \cite{genzel2023unified} therefore concluded that transition from nonuniform recovery to uniform recovery may lead to a slower error rate under discontinuous links. In this paper, we analyze the capability of PGD and generalized Lasso in uniform recovery.  We narrow the above gap by showing that  the uniform error rate is at most logarithmically worse than the nonuniform error rate under a large class of piecewise Lipschitz links. We further refine the analysis for iterative hard thresholding (a specific instance of PGD) under the sign link function, showing that there is no loss of logarithmic factor in this regime. In the following, we conclude with some discussions on the RAIC and the robustness under noise. 

\paragraph{Uniform Recovery via RAIC.} In fact, RAIC   gives a unified approach to uniform recovery of structured signals under nonlinear observations, consisting of the following steps:
\begin{enumerate}
    \item Construct a gradient map $\bh_{\bx}:\mathbb{R}^n\to \mathbb{R}^n$ for every $\bx\in\calX$ and recover $\bx$ via PGD
    \[\bx_{t+1}=P_{\calK}(\bx_t- \bh_{\bx}(\bx_t)),\quad t=0,1,\cdots.\]
    \item Identify an RAIC for $\bh_{\bx}$ that is enough to imply the convergence of the PGD to $\bx$ up to some statistical error.

    \item Establish this RAIC uniformly for all $\bx\in\calX$ (referred to as the uniform RAIC), which then implies the uniform recovery guarantee of PGD. See, e.g., Equations (\ref{uraicconvex}) and (\ref{uraiccone}).
\end{enumerate}

As in \cite[Section 1]{genzel2023unified}, the authors attributed the difficulty of \emph{uniform recovery under nonlinear observations} to a lack of technical tools such as restricted isometry property in nonlinear models, in light of a well-known fact from compressed sensing theory (see, e.g., \cite{Foucart2013AMI}) that RIP implies uniform sparse recovery under linear measurements. The above approach can be viewed as an extension of this fact for nonlinear observations, as the RAIC of the gradient of the squared $\ell_2$ loss---under linear measurements---precisely reduces to a certain form of the RIP (we leave a validation to avid reader).

    In the following, we provide a brief review for the recent works that analyzed PGD-type algorithms via RAIC. To our best knowledge, \cite{friedlander2021nbiht} first formally introduced the RAIC to analyze the normalized binary iterative hard thresholding (NBIHT) for one-bit compressed sensing, establishing the nonuniform guarantee 
    \(\|\hat{\bx}_{\rm nbiht}-\bx\|_2 =\tilde{O}\big(k^{7/2}/m\big)\)
    which is optimal in measurement number $m$ but suboptimal in the sparsity $k$. The work \cite{matsumoto2024binary} substantially improved the analysis and established the near-optimal uniform recovery guarantee 
    \begin{align}
        \sup_{\bx\in\Sigma^{n,*}_k}\|\hat{\bx}_{\rm nbiht}-\bx\|_2 = \tilde{O}(k/m) \label{mm24nbiht}
    \end{align}
        The analysis of \cite{matsumoto2024binary}  was   extended to general signal structures in \cite{chen2024optimal}, covering low-rank matrices and approximately sparse vectors. Moreover, RAIC was used to analyze other nonlinear observation models, including (sparse) logistic regression \cite{matsumoto2025learning,chen2026finite} and one-bit phase retrieval \cite{chen2024one}. We emphasize that, since these works used RAIC to analyze specific models, their formulations of the RAIC may exhibit some differences. More recently, \cite{chen2025unified} showed that RAIC   implies the local convergence of Riemannian gradient descent for the estimation of low Tucker rank tensor, while \cite{chen2026instance} showed that a RAIC could be used to establish instance optimal sparse recovery guarantees under nonlinear observations. Unlike previous works, \cite{chen2025unified,chen2026instance} defined the RAIC in a unifying manner and developed unified frameworks that readily apply to different models such as single-index models, logistic regression, ReLU regression, one-bit sensing and phase retrieval. We note that our RAIC in (\ref{uraiccone}) is a specific instance of \cite[Definition B.1]{chen2025unified}, and can be denoted by 
        \[\bh_{\bx}(\bu)\sim {\rm RAIC}\bigg((\calK-\calK)_{[1]};\calK,\mu\bx,\frac{1}{4}\|\bu-\mu\bx\|_2+B\bigg),\qquad \forall \bx\in\calX.\]
        under the conventions therein.

We mention that some of these works obtained uniform guarantees by (implicitly) following the above approach. For instance, \cite{matsumoto2024binary} derived (\ref{mm24nbiht}) by essentially establishing the uniform RAIC for all $\bx\in\Sigma^{n,*}_k$. See further examples in \cite{chen2024optimal,chen2026instance}. The present paper is also an instantiation of the above approach to Gaussian single-index models, with additional subtleties arising in bounding empirical processes, formulating the class of discontinuous links, and so on. 

\paragraph{Robustness.} Is PGD robust under noise or corruption? Here, we show that such robustness property can be readily incorporated into the RAIC analysis.  Let $\bh_{\bx}(\bu)$ be the gradient map in the \emph{noiseless} case, then the RAIC of $\bh_{\bx}(\bu)$ yields the convergence. The only impact of noise (or corruption) on this approach is that one can now only use the noisy (or corrupted) data to construct the gradient, say, $\tilde{\bh}_{\bx}(\bu)$, which can be essentially viewed as a \emph{perturbation} of the noiseless gradient  $\bh_{\bx}(\bu)$. The RAIC of $\tilde{\bh}_{\bx}(\bu)$ then, as before, implies the convergence of
\[\bx_{t+1}= P_{\calK}(\bx_t-\tilde{\bh}_{\bx}(\bx_t)),\quad t=0,1,\cdots.\]

In general, built upon the RAIC of $\bh_{\bx}(\bu)$, establishing the RAIC of $\tilde{\bh}_{\bx}(\bu)$ boils down to \emph{bounding an additional empirical process, which captures the effect of noise.} 
We demonstrate the idea through Gaussian single-index models. Since the derivation of robustness is somewhat orthogonal to the pursuit of uniform recovery, we shall start with nonuniform recovery of a fixed signal. We consider the setting of Theorem \ref{thm:recoverycone}, where $\bx \in\calK \cap \mathbb{S}^{n-1}$ for some cone $\calK$, the noiseless observations are $\{y_i=f_i(\ba_i^\top\bx)\}_{i=1}^m$, and Lemma \ref{lem:conepgdraic} states that the RAIC of the noiseless gradient $\bh_{\bx}(\bu)=\frac{1}{m}\bA^\top(\bA\bu-\bm{f}(\bA\bx))$, that is
\begin{align}
    \big\|\bu-\mu\bx- \bh_{\bx}(\bu)\big\|_{(\calK-\calK)^\circ_{[1]}}\le \frac{1}{4}\|\bu-\mu\bx\|_2 + B,\quad \forall \bu\in\calK, \label{nonuraicconecase}
\end{align}
implies the linear convergence of PGD to $O(B)$ statistical error. We now move forward to the noisy (or corrupted) setting where one observes 
\(\tilde{\by} = \bm{f}(\bA\bx)+\bm{\epsilon}\) for some $\bm{\epsilon}\in \mathbb{R}^n$, and is only able to construct the ``noisy'' gradient map
\[\tilde{\bh}_{\bx}(\bu) = \frac{1}{m}\bA^\top (\bA\bu-\tilde{\by}) = \bh_{\bx}(\bu)-\frac{1}{m}\bA^\top\bm{\epsilon}.\]
Built upon (\ref{nonuraicconecase}), $\tilde{\bh}_{\bx}(\bu)$ also satisfies the RAIC
\[ \big\|\bu-\mu\bx- \tilde{\bh}_{\bx}(\bu)\big\|_{(\calK-\calK)^\circ_{[1]}}\le \frac{1}{4}\|\bu-\mu\bx\|_2 + B + \bigg\|\frac{\bA^\top\bm{\epsilon}}{m}\bigg\|_{(\calK-\calK)^\circ_{[1]}},\quad \forall \bu\in\calK,\]
which implies the linear convergence of PGD (with $\tilde{\bh}_{\bx}(\bu)$) to $O\Big(B + \big\|\frac{\bA^\top\bm{\epsilon}}{m}\big\|_{(\calK-\calK)^\circ_{[1]}}\Big)$ error. As such, the noise vector $\bm{\epsilon}=(\epsilon_1,\cdots,\epsilon_m)^\top$ increases the error by 
\begin{align}
    \bigg\|\frac{\bA^\top\bm{\epsilon}}{m}\bigg\|_{(\calK-\calK)^\circ_{[1]}} = \bigg\|\frac{1}{m}\sum_{i=1}^m\epsilon_i\ba_i\bigg\|_{(\calK-\calK)^\circ_{[1]}}.\label{nonurobust}
\end{align}
An explicit control on this empirical process hinges on the assumption on $\bm{\epsilon}$. Assume $m\gtrsim \omega^2((\calK-\calK)_{[1]})$. If $\bm{\epsilon}$ is an arbitrary noise vector in $\mathbb{R}^m$, then  \[\bigg\|\frac{\bA^\top\bm{\epsilon}}{m}\bigg\|_{(\calK-\calK)^\circ_{[1]}}\lesssim \frac{\|\bm{\epsilon}\|_2}{\sqrt{m}}\] holds with high probability; if instead we consider   Gaussian noise $\bm{\epsilon}\sim N(0,\sigma^2\bI_m)$ independent of $\bA$, then with high probability, one has \[\bigg\|\frac{\bA^\top\bm{\epsilon}}{m}\bigg\|_{(\calK-\calK)^\circ_{[1]}}\lesssim \sigma \sqrt{\frac{ \omega^2((\calK-\calK)_{[1]})}{m}}.\] In fact, under the more fine-grained assumption on $\bm{\epsilon}$ as in \cite[Equation (2.4)]{genzel2023unified}, one can establish robustness property for the PGD that is comparable to \cite[Theorem 2]{genzel2023unified}.

It is straightforward to put  the robustness and uniform recovery together. In particular,  suppose that  for every $\bx$, one observes $\tilde{\by}=\bm{f}(\bA\bx)+\bm{\epsilon}_{\bx}$, then the noise increases the uniform recovery error by 
\[O\bigg(\sup_{\bx\in\calX}\,\bigg\|\frac{\bA^\top\bm{\epsilon}_{\bx}}{m}\bigg\|_{(\calK-\calK)^\circ_{[1]}} \bigg),\]
which exactly reduces to (\ref{nonurobust}) if the noise is signal-independent (i.e., $\bm{\epsilon}_{\bx}=\bm{\epsilon}$, for all $\bx$).
On a related note, \cite{matsumoto2024robust} extended \cite{matsumoto2024binary} to a corrupted setting with a small fraction of adversarial bit flips. The analysis of    \cite{matsumoto2024robust} is essentially an instantiation of the proposed approach to robustness of the PGD. 

\subsubsection*{Acknowledgment}
The authors used AI for polishing the language only. Part of the work was done when J.C. was affiliated with the Department of Mathematics, University of Maryland, College Park.

\appendix

\section{Proof of Lemma \ref{lem:smallball}}\label{app:prooflem22}
Fix a realization of $f$, and let $\varphi$ denote the standard
Gaussian density. If $|\calD_f|\le 1$, the conclusion follows
immediately from $\varphi(x)\le 1/\sqrt{2\pi}$.
We therefore assume that $\beta<\infty$. We first show that, for any $\beta$-separated points
$x_1<\cdots<x_N$,
\begin{equation}\label{eq:separated-gaussian-density}
    \sum_{j=1}^N\varphi(x_j)
    \le \frac{1}{\beta}+\frac{1}{\sqrt{2\pi}}.
\end{equation}
Choose $k$ such that $\varphi(x_k)=\max_{1\le j\le N}\varphi(x_j)$.
By the unimodality of $\varphi$, we have
\[
    \beta\varphi(x_j)
    \le
    \begin{cases}
        \displaystyle\int_{x_j}^{x_{j+1}}\varphi(x)\,dx,
        & j<k,\\[2mm]
        \displaystyle\int_{x_{j-1}}^{x_j}\varphi(x)\,dx,
        & j>k.
    \end{cases}
\]
Summing these inequalities gives
$\beta\sum_{j\ne k}\varphi(x_j)\le 1$.
Together with $\varphi(x_k)\le 1/\sqrt{2\pi}$, this proves
\eqref{eq:separated-gaussian-density}. Now let $D=\{\xi_1,\ldots,\xi_N\}$ be any nonempty finite subset
of $\calD_f$. Since $\{\xi_j+u\}_{j=1}^N$ is $\beta$-separated
for every $u\in\mathbb{R}$, the union bound and
\eqref{eq:separated-gaussian-density} yield
\begin{align*}
    \mathbb{P}\bigl(\dist(g,D)\le t\bigr)
    \le \sum_{j=1}^N\int_{-t}^{t}\varphi(\xi_j+u)\,du=\int_{-t}^{t}\sum_{j=1}^N\varphi(\xi_j+u)\,du \le \bigg(\frac{2}{\beta}
        +\sqrt{\frac{2}{\pi}}\bigg)t.
\end{align*}
If $\calD_f$ is infinite, apply this bound to the finite sets $\calD_f\cap[-M,M]$ and let $M\to\infty$.

\section{Proofs of Lemmas \ref{lem:coreconvex} and \ref{lem:corecone}} \label{app:coreprocess}
In the following, we simultaneously bound both empirical processes by bounding 
\[\sup_{\bx\in\calX}\sup_{\bv\in \calV}\,\frac{1}{m}\sum_{i=1}^m\big(f_i(\ba_i^\top\bx)-\mu\ba_i^\top\bx\big)\ba_i^\top\bv\]
for some $\calV\subset \mathbb{B}_2^n$ by a covering argument. While the main ideas of the covering argument have already appeared in existing analysis for specific nonlinear links (see, e.g., \cite{xu2020quantized,chen2024robust}), adaptations are needed for treating the general links in Assumption \ref{assump:fi}.

 \begin{lemma}[Covering bound] \label{lem:uniformmultiplier}
Under Assumptions \ref{assump:tsim} and \ref{assump:fi}, let
$\calX\subset\mathbb{S}^{n-1}$ and $\calV\subset\mathbb{B}_2^n$. 
For $r\in(0,1)$, $t\ge 1$ and
$\zeta\in\{1/m,2/m,\cdots,1\}$,  define
\begin{align}\label{ump:delta}
    \delta:=\zeta+\frac{\log\calN(\calX;r)+t^2}{m}
    +\bigg[\frac{2}{\beta}+\sqrt{\frac{2}{\pi}}\bigg]\bigg(
        \frac{\omega_r(\calX-\calX)}{\sqrt{\zeta m}}
        +r\sqrt{\log\frac{e}{\zeta}}
    \bigg).
\end{align}
If $\delta\le c$ for a sufficiently small universal constant $c>0$
and $t^2\le\zeta m\log(e/\zeta)$, then with probability at least
$1-C\exp(-ct^2)$ over the randomness of $\{\ba_i\}_{i=1}^m$ and
$\{f_i\}_{i=1}^m$,
\begin{align}
    &\sup_{\bx\in\calX}\sup_{\bv\in\calV}
    \bigg|\frac{1}{m}\sum_{i=1}^m
    \big(f_i(\ba_i^\top\bx)-\mu\ba_i^\top\bx\big)
    \ba_i^\top\bv\bigg|\notag\\
    &\lesssim
    \frac{\sigma\big(\omega(\calV)+\sqrt{\log\calN(\calX;r)+t^2}\big)}{\sqrt m}
    \\&\label{term11}+(\mu+\gamma)
    \bigg(\frac{\omega_r(\calX-\calX)}{\sqrt m}+r\bigg)
    \bigg(\frac{\omega(\calV)}{\sqrt m}+1\bigg)
    \\&+\tau\sqrt{\delta}
    \bigg(\frac{\omega(\calV)}{\sqrt m}
        +\sqrt{\delta\log\frac{e}{\delta}}\bigg).\label{ump:bound}
\end{align} 
If, in addition, 
$m\gtrsim\omega^2(\calV)+\omega_1^2(\frac{\calX-\calX}{r})$,
then the term in (\ref{term11}) is bounded by $O((\mu+\gamma) r)$. 
\end{lemma}

\begin{proof}
In this proof, we let $\kappa:=2/\beta+\sqrt{2/\pi}.$
Write $\tilde f_i(b):=f_i(b)-\mu b$ and $L:=\mu+\gamma$. By Assumption \ref{assump:fi}, 
the functions $\tilde f_i$ are piecewise $L$-Lipschitz,
with jumps bounded by $\tau$ and separated by at least $\beta$.
Let $\calS\subset\calX$ be an $r$-net with
$|\calS| = \calN(\calX;r)$. For every $\bx\in\calX$, choose
$\hat{\bx}\in\calS$ such that $\|\bx-\hat{\bx}\|_2\le r$ (ties are broken arbitrarily).
In particular, $\bx-\hat{\bx}\in\calU:=(\calX-\calX)_{[r]}$.

\paragraph{Bounding the process on the net.}
Write $Z_{\bx,\bv}:=m^{-1}\sum_{i=1}^m
\tilde f_i(\ba_i^\top\bx)\ba_i^\top\bv$
for $\bx\in\calX$ and $\bv\in\calV$.
By $\calX\subset\mathbb{S}^{n-1}$, Gaussian rotational invariance
and (\ref{mufgg}),
\begin{align*}
    \mathbb{E}\big[
    \tilde f_i(\ba_i^\top\bx)\ba_i^\top\bv\big]
    = \mathbb{E}\big[f_i(\ba_i^\top\bx)\ba_i^\top\bv\big] - \mu\bx^\top\bv=\mathbb{E}\big[f_i(\ba_i^\top\bx)\ba_i^\top\bx\big] (\bx^\top\bv)  - \mu\bx^\top\bv = 0.
\end{align*}
For any fixed $\bp\in\calS$, (\ref{sigmafgmg}) gives
$\|\tilde f_i(\ba_i^\top\bp)\|_{\psi_2}=\sigma$.
Apply Lemma \ref{menproduct} with the first index set being
$\{\bp\}$.
Here the first process has no increments, while
$\|\ba_i^\top(\bv-\bv')\|_{\psi_2}\lesssim\|\bv-\bv'\|_2$
and $\sup_{\bv\in\calV}\|\ba_i^\top\bv\|_{\psi_2}\lesssim 1$.
Thus, for $1\le u\le\sqrt m$, with probability at least
$1-2\exp(-cu^2)$,
\begin{align*}
    \sup_{\bv\in\calV}|Z_{\bp,\bv}|
    &\lesssim\sigma\bigg(
        \frac{u(\omega(\calV)+u)}{m}
        +\frac{\omega(\calV)+u}{\sqrt m}\bigg)
    \lesssim\frac{\sigma(\omega(\calV)+u)}{\sqrt m}.
\end{align*}
Since $(\log\calN(\calX;r)+t^2)/m\le\delta\le c$, we may take
$u=C\sqrt{\log\calN(\calX;r)+t^2}\le\sqrt m$ for a sufficiently large universal
$C$. A union bound over $\bp\in\calS$ then gives
\begin{align}\label{ump:net}
    \sup_{\bp\in\calS}\sup_{\bv\in\calV}|Z_{\bp,\bv}|
    \lesssim
    \frac{\sigma\big(\omega(\calV)+\sqrt{\log\calN(\calX;r)+t^2}\big)}{\sqrt m}
\end{align}
with probability at least $1-2\exp(-ct^2)$. In view of 
\begin{align}
    \label{passnet}\sup_{\bx\in\calX}\sup_{\bv\in\calV}|Z_{\bx,\bv}|\le\sup_{\bx\in\calX}\sup_{\bv\in\calV}|Z_{\hat{\bx},\bv}|+\sup_{\bx\in\calX}\sup_{\bv\in\calV}|Z_{\bx,\bv}-Z_{\hat{\bx},\bv}|\le    \sup_{\bp\in\calS}\sup_{\bv\in\calV}|Z_{\bp,\bv}| +\sup_{\bx\in\calX}\sup_{\bv\in\calV}|Z_{\bx,\bv}-Z_{\hat{\bx},\bv}|, 
\end{align}
it remains to bound
$G:=\sup_{\bx\in\calX}\sup_{\bv\in\calV}
|Z_{\bx,\bv}-Z_{\hat{\bx},\bv}|$.

\paragraph{Controlling the number of exceptional measurements.}
For a deterministic set $\mathcal{T}\subset\mathbb{R}^n$ and
$\rho\in\{1/m,2/m,\cdots,1\}$, write
\begin{align*}
    M_{\mathcal{T}}(\rho):=
    \sup_{\bw\in\mathcal{T}}
    \max_{\substack{\calI\subset[m]\\|\calI|\le\rho m}}
    \bigg(\frac{1}{m}\sum_{i\in\calI}
    |\ba_i^\top\bw|^2\bigg)^{1/2}.
\end{align*}
Padding the index set allows us to apply Lemma \ref{lem:maxlsum}
with $|\calI|\le\rho m$ instead of $|\calI|=\rho m$.
Choose
\begin{align}\label{ump:eta}
    \eta=C\bigg(
        \frac{\log\calN(\calX;r)+t^2}{\kappa m}
        +\frac{\omega_r(\calX-\calX)}{\sqrt{\zeta m}}
        +r\sqrt{\log\frac{e}{\zeta}}
    \bigg)
\end{align}
for a sufficiently large universal $C$.
For $\bp\in\calS$ and $\bx\in\calX$, define
\begin{align*}
    \calJ_{\bp}
    &:=\{i\in[m]:\dist(\ba_i^\top\bp,\calD_{f_i})\le\eta\},\\
    \calI_{\bx}
    &:=\{i\in[m]:|\ba_i^\top(\bx-\hat{\bx})|\ge\eta/2\},
    \\
    \calG_{\bx}&:=\calJ_{\hat{\bx}}\cup\calI_{\bx}.
\end{align*} 
By Lemma \ref{lem:smallball} and the independence of $f_i$ and
$\ba_i$, we have \[\mathbb{P}(i\in\calJ_{\bp})\le\kappa\eta.\] 
The indicators are independent across $i$, so a Chernoff bound yields
\[\mathbb{P}(|\calJ_{\bp}|>2\kappa\eta m)
\le\exp(-c\kappa\eta m).\] 
Since (\ref{ump:eta}) ensures
$\kappa\eta m\ge C(\log\calN(\calX;r)+t^2)$, a union bound over $\bp\in\calS$
gives \[\mathbb{P}\bigg(\sup_{\bp\in\calS}|\calJ_{\bp}|\le 2\kappa\eta m\bigg)\ge 1-\exp(-ct^2).\] 
Next, Lemma \ref{lem:maxlsum} applied to $\calU$ gives
\begin{align}\label{ump:largeincrements}
    \frac{M_{\calU}(\zeta)}{\sqrt\zeta}
    \le C\bigg(
    \frac{\omega_r(\calX-\calX)}{\sqrt{\zeta m}}
    +r\sqrt{\log\frac{e}{\zeta}}\bigg)
    \le\frac{\eta}{4}
\end{align}
with probability at least $1-2\exp(-c\zeta m\log(e/\zeta))$.
On this event, $|\calI_{\bx}|\le\zeta m$ holds for all
$\bx\in\calX$. Indeed, otherwise choosing $\zeta m$ indices
from $\calI_{\bx}$ and using $\bx-\hat{\bx}\in\calU$ would
give $M_{\calU}(\zeta)/\sqrt\zeta\ge\eta/2$,
contradicting (\ref{ump:largeincrements}).
Together with the bound on $\calJ_{\bp}$, this yields
\[|\calG_{\bx}|\le(2\kappa\eta+\zeta)m,\qquad \forall \bx\in\calX.\]
Define the deterministic fraction
$q:=\lceil(2\kappa\eta+\zeta)m\rceil/m$.
Since $1/m\le\zeta$ and $\kappa\eta\lesssim\delta$, choosing
$c$ in the condition $\delta\le c$ sufficiently small ensures
\begin{align}\label{ump:badfraction}
    q\in\{1/m,2/m,\cdots,1\},\qquad
    \zeta\le q\lesssim\delta<1,\qquad
    |\calG_{\bx}|\le qm,\quad\forall\bx\in\calX.
\end{align}

\paragraph{Extending the bound from the net to $\calX$.}
By piecewise Lipschitz continuity and the bound on the jumps,
\begin{align}\label{ump:globalincrement}
    |\tilde f_i(b)-\tilde f_i(c)|
    \le L|b-c|+\tau\bigg(\frac{|b-c|}{\beta}+1\bigg),
    \qquad b,c\in\mathbb{R}.
\end{align}
To see this, the interval between $b$ and $c$ contains at most
$1+\frac{|b-c|}{\beta}$ discontinuities. This, together with the convention $f(b)\in\{\lim_{x\to b^-}f(x),\lim_{x\to b^+}f(x)\}$ when $b\in \calD_f$, leads to the jump term $\tau(\frac{|b-c|}{
\beta
}+1)$. The term $L|b-c|$ is the Lipschitz term.

If $i\notin\calG_{\bx}$, then
$\dist(\ba_i^\top\hat{\bx},\calD_{f_i})>\eta$ and
$|\ba_i^\top(\bx-\hat{\bx})|<\eta/2$.
Thus, the closed interval between $\ba_i^\top\bx$ and
$\ba_i^\top\hat{\bx}$ contains no discontinuity, and only the
Lipschitz term is needed. 
Combining this observation with
(\ref{ump:globalincrement}), for every $i\in[m]$ we have
\begin{align*}
    |\tilde f_i(\ba_i^\top\bx)
        -\tilde f_i(\ba_i^\top\hat{\bx})| \le L|\ba_i^\top(\bx-\hat{\bx})|
    +\tau\bigg(
        \frac{|\ba_i^\top(\bx-\hat{\bx})|}{\beta}+1
    \bigg)\1_{\{i\in\calG_{\bx}\}}.
\end{align*}
Multiply by $|\ba_i^\top\bv|$, sum over $i$, and use
Cauchy--Schwarz together with (\ref{ump:badfraction}).
Since $\bx-\hat{\bx}\in\calU$, we obtain
\begin{align}\label{ump:incrementbound}
    G\le L M_{\calU}(1)M_{\calV}(1)
    +\tau\bigg(\frac{M_{\calU}(q)}{\beta}+\sqrt q\bigg)
    M_{\calV}(q).
\end{align}
Apply Lemma \ref{lem:maxlsum} to both $\calU$ and $\calV$,
with fractions $q$ and $1$. Simultaneously for $\rho\in\{q,1\}$,
\begin{align}\label{ump:restrictednorms}
    M_{\calU}(\rho)
    &\lesssim\frac{\omega_r(\calX-\calX)}{\sqrt m}
        +r\sqrt{\rho\log\frac{e}{\rho}},\notag\\
    M_{\calV}(\rho)
    &\lesssim\frac{\omega(\calV)}{\sqrt m}
        +\sqrt{\rho\log\frac{e}{\rho}}
\end{align}
except on an event of probability at most
$C\exp(-c\zeta m\log(e/\zeta))$.
Here we used $\rho\ge\zeta$ and the monotonicity of
$\rho\log(e/\rho)$ on $(0,1]$.
In particular, the contribution from $\rho=1$ is
\begin{align}\label{ump:lipschitzpart}
    L M_{\calU}(1)M_{\calV}(1)
    \lesssim(\mu+\gamma)
    \bigg(\frac{\omega_r(\calX-\calX)}{\sqrt m}+r\bigg)
    \bigg(\frac{\omega(\calV)}{\sqrt m}+1\bigg).
\end{align}
We finally bound the jump contribution in (\ref{ump:incrementbound}).
By $q\lesssim\delta\le c$,
$q\log(e/q)\lesssim\delta\log(e/\delta)$.
Moreover, $\kappa\ge \frac{2}{\beta}$ and (\ref{ump:delta}) imply
\begin{align*}
    \frac{M_{\calU}(q)}{\beta}
    &\lesssim
        \frac{\omega_r(\calX-\calX)}{\beta\sqrt m}
        +\frac{r}{\beta}\sqrt{\delta\log\frac{e}{\delta}}\lesssim\delta\sqrt\zeta
        +\delta\sqrt\delta\,
        \sqrt{\frac{\log(e/\delta)}{\log(e/\zeta)}}
    \lesssim\delta^{3/2}\lesssim\sqrt\delta,
\end{align*}
where we used $\zeta\le\delta<1$ in the last line.
The argument includes $\beta=\infty$ under the convention
$\frac{b}{\infty}=0$ for any $b\in \mathbb{R}$. 
Consequently, (\ref{ump:restrictednorms}) yields
\begin{align}\label{ump:jumppart}
    \tau\bigg(\frac{M_{\calU}(q)}{\beta}+\sqrt q\bigg)
    M_{\calV}(q)
    \lesssim\tau\sqrt\delta\bigg(
        \frac{\omega(\calV)}{\sqrt m}
        +\sqrt{\delta\log\frac{e}{\delta}}\bigg).
\end{align}
Combining (\ref{ump:net}), (\ref{passnet}), (\ref{ump:incrementbound}),
(\ref{ump:lipschitzpart}) and (\ref{ump:jumppart}) proves
the claimed bound. Since
$t^2\le\zeta m\log(e/\zeta)$,
all the required events hold simultaneously with probability
at least $1-C\exp(-ct^2)$. 
The further simplification under $m\gtrsim \omega^2(\calV)+\omega_1^2(\frac{\calX-\calX}{r})$ is straightforward. The proof is complete.  
\end{proof}

\section{Proof of Theorem \ref{thm:sharp1b}} \label{app:1bcs1prove}
In light of Lemma \ref{lem:coneRAIC-reduce}, it remains to establish  
\begin{align*}
\sup_{\bx\in\Sigma^{n,*}_k}\sup_{\bq\in\Sigma^{n,*}_{2k}}\,\frac{1}{m}\sum_{i=1}^m \bigg(\sign(\ba_i^\top \bx)-\sqrt{\frac{2}{\pi}}\ba_i^\top \bx\bigg)\ba_i^\top \bq \lesssim \sqrt{\frac{k\log(en/k)}{m}} 
\end{align*}
with the promised probability under $m\gtrsim k\log(en/k)$.  To this end, we decouple the process into two centered processes  
\begin{equation*}
\sqrt{\frac{2}{\pi}} \underbrace{\sup_{\bx\in\Sigma^{n,*}_k}\sup_{\bq\in\Sigma^{n,*}_{2k}}\frac{1}{m}\sum_{i=1}^m \bigg(\bx^\top \bq-(\ba_i^\top \bx)(\ba_i^\top \bq) \bigg)}_{\coloneqq \Xi_2} + \underbrace{\sup_{\bx\in\Sigma^{n,*}_k}\sup_{\bq\in\Sigma^{n,*}_{2k}}\frac{1}{m}\sum_{i=1}^m \bigg(\sign(\ba_i^\top \bx)\ba_i^\top \bq -\sqrt{\frac{2}{\pi}}\bx^\top \bq\bigg)}_{\coloneqq \Xi_3}.
\end{equation*} 
By Lemma \ref{menproduct}, we obtain that $\Xi_2\lesssim\sqrt{k\log(en/k)/m}$ with the promised probability. It remains to show $\Xi_3\lesssim \sqrt{k\log(en/k)/m}$.

\subsection{Controlling $\Xi_3$}
We would like to proceed similarly to control the process $\Xi_3$. However, the sign function is not a class of function with subgaussian increments. To bypass the lack of regularity of the sign function, we rely on a covering argument exploiting the metric entropy estimates for Boolean classes of functions that exploit the VC dimension of the class. We refer the reader to \cite{vershynin2018high} for a background on VC dimension, but it will not be strictly necessary here as the estimates based on VC-dimension used here are well-known facts in the literature. To simplify the notation, we shall embed  $\bx\in \Sigma^{n,*}_{k}\subset \Sigma^{n,*}_{2k}$ and seek to bound 
\[\sup_{\bx,\bq\in\Sigma^{n,*}_{2k}}\bigg|\frac{1}{m}\sum_{i=1}^m \sign(\ba_i^\top \bx)\ba_i^\top \bq-\sqrt{\frac{2}{\pi}}\bx^\top \bq\bigg|.\]  

To start, we reduce the control of $\Xi_3$ to its symmetrized version. More accurately, we apply the Gine-Zinn symmetrization for tail estimates (see, e.g., \cite[Theorem 1.14]{mendelson2016upper}) to obtain that for any $u\ge 2\sqrt{\pi/m}$
\begin{equation}
\label{eq:symmetrization_step}
\begin{split}
&\mathbb{P}\left(\sup_{\bx,\bq\in \Sigma^{n,*}_{2k}} \bigg|\frac{1}{m}\sum_{i=1}^m \sign(\ba_i^\top \bx)\ba_i^\top \bq-\sqrt{\frac{2}{\pi}}\bx^\top \bq\bigg|\ge \sqrt{\frac{2}{\pi}}u\right)\\
&\le 4 \mathbb{P}\left(\sup_{\bx,\bq\in \Sigma^{n,*}_{2k}} \bigg|\frac{1}{m}\sum_{i=1}^m \varepsilon_i \sign(\ba_i^\top \bx)\ba_i^\top \bq \bigg|\ge \frac{u}{4}\sqrt{\frac{2}{\pi}}\right),
\end{split}
\end{equation}
where $\varepsilon_1,\cdots,\varepsilon_m$ are random signs independent from $\ba_1,\cdots,\ba_m$. Since the law of such signs and the index set are symmetric, we may drop the absolute value and focus on the process $Z_{\bx,\bq}$ given by
\begin{equation*}
 Z_{\bx,\bq}\coloneqq\frac{1}{\sqrt{m}}\sum_{i=1}^m \varepsilon_i \sign(\ba_i^\top \bx)\ba_i^\top \bq, \quad (\bx,\bq)\in \Sigma^{n,*}_{2k}\times \Sigma^{n,*}_{2k} .
\end{equation*}
We require a technical lemma that brings together some technical results in the literature. To state them accurately, let $\mathcal{N}(T,d,\varepsilon)$ be covering number of $T$ with respect to $d$ at the scale of $\varepsilon$, that is, the smallest number of balls of radius $\varepsilon$ with respect to the (pseudo) metric $d$ whose centers lie in $T$. 

\begin{lemma}
\label{lemmma:covering}
Let $\mathcal{X}=\{\by_1,\cdots, \by_m\}$ be any collection of points in $\mathbb{R}^n$ and $\mu$ be any probability measure on the set $[m]$. Consider  $\|\cdot\|_{\mathcal{X},\mu}$ to be the following (pseudo) metric in $\mathbb{R}^n$ 
\begin{equation*}
\|\bx-\bx'\|_{\mathcal{X},\mu}\coloneqq  \left(\sum_{i=1}^m\mu(\{i\})\big|\sign(\by_i^\top \bx)-\sign(\by_i^\top \bx')\big|^2\right)^{1/2}, \quad  \bx,\bx'\in \mathbb{R}^n.
\end{equation*}
Then, for any $\varepsilon \in (0,2)$, 
\begin{equation*}
\log \mathcal{N}(\Sigma^{n,*}_{2k},\|\cdot\|_{\mathcal{X},\mu},\varepsilon)\lesssim k\log\bigg(\frac{en}{k}\bigg)\log\bigg(\frac{4}{\varepsilon}\bigg).
\end{equation*}
\end{lemma}
\begin{proof}
We first observe that 
\begin{equation*}
|\sign(\by^\top \bx)-\sign(\by^\top \bx')| = 2 |\mathbbm{1}(\langle \by,\bx\rangle\ge 0)-\mathbbm{1}(\langle \by,\bx'\rangle\ge 0)|
\end{equation*}
Consequently, up to a factor of $2$, it suffices to estimate the metric entropy of the boolean class of function $\mathcal{F}=\{\mathbbm{1}(\langle \cdot,\bx\rangle \ge 0):\bx\in \Sigma^{n,*}_{2k} \}$ with respect to the $L^2$ metric endowed by $\mu$, that is,
\begin{equation*}
\|f-f'\|_{L^2(\mathcal{X},\mu)}
\coloneqq
\left(\sum_{i=1}^m \mu(\{i\})\big((f-f')(y_i)\big)^2\right)^{1/2},
\quad f,f'\in\mathcal{F}.
\end{equation*}
By the metric entropy estimate for VC class (see for  \cite[Theorem 8.3.18]{vershynin2018high}), we obtain that  
\begin{equation*}
\log \mathcal{N}(\Sigma^{n,*}_{2k},\|\cdot\|_{\mathcal{X},\mu},\varepsilon)= \log \mathcal{N}(\mathcal{F},\|\cdot\|_{L^2(\mathcal{X},\mu)},\varepsilon/2) \lesssim vc(\mathcal{F})\log\bigg(\frac{4}{\varepsilon}\bigg),
\end{equation*}
where $vc(\mathcal{F})$ is the VC-dimension of the class $\mathcal{F}$. The proof now follows by combining the estimate above with $vc(\mathcal{F}) \lesssim k\log(en/k)$ (see \cite[Corollary 2]{depersin2024robust}). 
\end{proof}

We now focus on obtaining the estimate for the supremum of the process $Z_{\bx,\bq}$. Notice that the metric entropy estimate in Lemma \ref{lemmma:covering} holds for any collection of points $\mathcal{X}$, thus the law of $\ba_1,\cdots,\ba_m$ does not affect the estimates there. This is the typical advantage of an approach relying on covering numbers of Boolean classes of functions.

To start, we condition on the Gaussian matrix $\bA\in\mathbb R^{m\times n}$,
whose rows are $\ba_1,\ldots,\ba_m$. Thus, $\{Z_{\bx,\bq}\}_{(\bx,\bq)\in \Sigma^{n,*}_{2k}\times \Sigma^{n,*}_{2k}}$  becomes a Bernoulli process whose subgaussian increments satisfy
\[
\|Z_{\bx,\bq}-Z_{\bx',\bq'}\|_{\psi_2}
\lesssim
\rho_{\bA}((\bx,\bq),(\bx',\bq')),
\]
where
\[
\rho_{\bA}((\bx,\bq),(\bx',\bq'))^2
=
\frac1m
\sum_{i=1}^m
\Big(
\sign(\langle \ba_i,\bx\rangle)\langle \ba_i,\bq\rangle
-
\sign(\langle \ba_i,\bx'\rangle)\langle \ba_i,\bq'\rangle
\Big)^2.
\]

Adding and subtracting
$\sign(\langle \ba_i,\bx\rangle)\langle \ba_i,\bq'\rangle$ inside the summand,
and applying the triangle inequality in the empirical $L_2$ norm (Minkowski inequality) yields
\begin{equation}
\label{eq:metric_decomposition}
\rho_{\bA}((\bx,\bq),(\bx',\bq'))
\le
\|\bq-\bq'\|_{\bA}
+
d_{\bA,\bq'}(\bx,\bx'),
\end{equation}
where
\[
\|\bq-\bq'\|_{\bA}^2
=
\frac1m
\sum_{i=1}^m
\langle \ba_i,\bq-\bq'\rangle^2 \quad \text{and} \quad d_{\bA,\bq}^2(\bx,\bx')
=
\frac1m
\sum_{i=1}^m
\Big(
\sign(\langle \ba_i,\bx\rangle)
-
\sign(\langle \ba_i,\bx'\rangle)
\Big)^2
\langle \ba_i,\bq\rangle^2.
\]
Next, define the event  
\[
\mathcal E
=
\left\{
\sup_{\bv\in \Sigma^{n,*}_{4k}}
\frac{\|\bA\bv\|_2}{\sqrt m}
\le 4
\right\},
\]
and notice that on $\mathcal E$, we have
\[
\|\bq-\bq'\|_{\bA}
\le
4\|\bq-\bq'\|_2,
\]
for every $\bq,\bq'\in \Sigma^{n,*}_{2k}$. Moreover, by the standard Gaussian restricted isometry estimate (see, e.g., \cite[Chapter 9]{vershynin2018high}),
\[
\mathbb P(\mathcal E^c)
\le
e^{-cm},
\]
provided $m\ge Ck\log\frac{en}{k}$ for large enough constant $C$.  Consequently,
\begin{equation}
\label{eq:final_prob_estimate}
\mathbb P\!\left(
\sup_{\bx,\bq}Z_{\bx,\bq}\ge t
\right)
\le
\mathbb P\!\left(
\left\{
\sup_{\bx,\bq}Z_{\bx,\bq}\ge t
\right\}
\cap
\mathcal E
\right)
+
e^{-cm}.
\end{equation}
It therefore suffices to estimate the process on the event $\mathcal E$. Our goal is to apply Dudley's inequality.

To this end, we need to estimate the covering numbers of $\Sigma^{n,*}_{2k}\times \Sigma^{n,*}_{2k}$ with respect to the metric $\rho_{\bA}$.
\noindent On $\mathcal E$, a Euclidean $\varepsilon/8$-net of $\Sigma^{n,*}_{2k}$ is also an $\varepsilon/2$-net with respect to $\|\cdot\|_{\bA}$. Hence, by the standard volumetric estimate

\[
\log
\mathcal N
(\Sigma^{n,*}_{2k},\|\cdot\|_{\bA},\varepsilon/2)
\lesssim
k\log\frac{en}{k}
+
k\log\frac{16}{\varepsilon}.
\]
Fix now $\bq\in \Sigma^{n,*}_{2k}$ and define the probability measure
\[
\mu(\{i\})
=
\frac{\langle \ba_i,\bq\rangle^2}
{\sum_{j=1}^m
\langle \ba_j,\bq\rangle^2},
\qquad
i=1,\ldots,m,
\]
and the Boolean function class $\mathcal{F}$ of functions $f_{\bx}(i)=\sign(\langle \ba_i,\bx\rangle)$ where $\bx\in \Sigma^{n,*}_{2k}$. Thus, on the event $\mathcal E$,
\[
d_{\bA,\bq}(\bx,\bx')
=
\left(
\frac1m
\sum_{i=1}^m
\langle \ba_i,\bq\rangle^2
\right)^{1/2}
\|f_{\bx}-f_{\bx'}\|_{L^2(\mu)}\le 4\|f_{\bx}-f_{\bx'}\|_{L^2(\mu)}.
\]
By Lemma \ref{lemmma:covering},
\[
\log
\mathcal N
(\Sigma^{n,*}_{2k},
d_{\bA,\bq},
\varepsilon/2)
\lesssim
k\log\frac{en}{k}
\log\frac{16}{\varepsilon}.
\]
Since 
\[
\mathcal{N}(\Sigma^{n,*}_{2k}\times \Sigma^{n,*}_{2k},\rho_{\bA},\varepsilon)\le \mathcal{N}(\Sigma^{n,*}_{2k},\|\cdot\|_{\bA},\varepsilon/2)\sup_{\bq\in \Sigma^{n,*}_{2k}}\mathcal{N}(\Sigma^{n,*}_{2k},d_{\bA,\bq},\varepsilon/2)
\]
Consequently, the Dudley's integral is estimated by
\[
\int_{0}^{8}\sqrt{\log
\mathcal N
(\Sigma^{n,*}_{2k}\times \Sigma^{n,*}_{2k},
\rho_{\bA},
\varepsilon)} d\varepsilon
\lesssim
\sqrt{k\log\frac{en}{k}}
\int_{0}^{8}\sqrt{\log\frac{16}{\varepsilon}}d\varepsilon\lesssim \sqrt{k\log\frac{en}{k}},
\]
Applying Dudley's inequality in tail form \cite{vershynin2018high}, there exists a constant $C$ for which 
\begin{align*}
\mathbb P\!\left(
\left\{
\frac{1}{\sqrt{m}}\sup_{\bx,\bq}Z_{\bx,\bq}\ge C\sqrt{\frac{k\log(en/k)}{m}}
\right\}
\right)&\le \mathbb P\!\left(
\left\{
\sup_{\bx,\bq}Z_{\bx,\bq}\ge C\sqrt{k\log(en/k)}
\right\}
\cap
\mathcal E
\right) + e^{-cm}\\&\lesssim e^{-ck\log(en/k)},
\end{align*}
provided that $m\ge Ck\log(en/k)$ for a sufficiently large constant $C>0$.

\section{Technical Lemmas}
\begin{lemma}[{\cite[Theorem 1.13]{mendelson2016upper}}, see also {\cite[Theorem 8]{genzel2023unified}}]  \label{menproduct}
Let $\ba_1,\dots,\ba_m$ be i.i.d. copies of a random vector $\ba$, and
let $\{g_{\bu}(\ba)\}_{\bu\in \calU}$ and $\{h_{\bv}(\ba)\}_{\bv\in\calV}$ be real-valued stochastic processes indexed by 
$\calU,\calV\subset \mathbb{R}^n$, respectively. Assume that there exist $K_1,K_2,r_1,r_2\geq 0$ such that
\begin{gather*}
\|g_{\bu}(\ba)-g_{\bu'}(\ba)\|_{\psi_2}\leq K_1\|\bu-\bu'\|_2,\quad
\|g_{\bu}(\ba)\|_{\psi_2} \leq r_1,\quad \forall\bu,\bu'\in \calU,\\
\|h_{\bv}(\ba)-h_{\bv'}(\ba)\|_{\psi_2} \leq K_2\|\bv-\bv'\|_2,\quad
\|h_{\bv}(\ba)\|_{\psi_2}\leq r_2,\quad \forall\bv,\bv'\in \calV. 
\end{gather*}
Then for every $t\geq 1$, with probability at least $1-2\exp(-ct^2)$, it holds that
\begin{align*}
&\sup_{\bu\in\calU}\sup_{\bv\in\calV} 
\bigg| \frac{1}{m}\sum_{i=1}^m g_{\bu}(\ba_i)h_{\bv}(\ba_i)
- \mathbb{E}\big[g_{\bu}(\ba)h_{\bv}(\ba)\big]\bigg|\\
&\leq C\bigg(\frac{(K_1\omega(\calU)+t r_1) \cdot(K_2 \omega(\calV)+t r_2)}{m}
+\frac{r_1 K_2\omega(\calV)+r_2 K_1\omega(\calU)+t r_1r_2}{\sqrt{m}}\bigg).
\end{align*}
\end{lemma}
 
 \begin{lemma}\cite[Theorem 2.10]{dirksen2021non} \label{lem:maxlsum}
 Let $\ba_1,\dots,\ba_m$ be i.i.d. copies of $N(0,I_n)$, and let $\calU\subset \mathbb{R}^n$.
      For any $\zeta \in \{\frac{1}{m},\frac{2}{m},\cdots,\frac{m-1}{m},1\}$, the event 
      \begin{align*}
          \sup_{\bu\in \calU}\max_{\substack{\calG\subset [m]\\|\calG|= \zeta m}}\,\bigg(\frac{1}{\zeta m}\sum_{i\in \calG}|\ba_i^\top \bu|^2\bigg)^{1/2} \lesssim \frac{\omega(\calU)}{\sqrt{\zeta m}} + \Big(\sup_{\bu\in\calU}\|\bu\|_2\Big)\sqrt{\log\frac{e}{\zeta}}
      \end{align*}
      holds with probability at least $1-2\exp(-c\zeta m\log(\frac{e}{\zeta}))$.
  \end{lemma}

\bibliography{libr}

\end{document}

%% file: preamble.tex
\newcommand{\ba}{\mathbf{a}}
\newcommand{\bA}{\mathbf{A}}

\newcommand{\bg}{\mathbf{g}}

\newcommand{\bh}{\mathbf{h}}

\newcommand{\bI}{\mathbf{I}}

\newcommand{\bp}{\mathbf{p}}

\newcommand{\bq}{\mathbf{q}}

\newcommand{\bs}{\mathbf{s}}

\newcommand{\bu}{\mathbf{u}}

\newcommand{\bv}{\mathbf{v}}

\newcommand{\bw}{\mathbf{w}}

\newcommand{\bx}{\mathbf{x}}

\newcommand{\by}{\mathbf{y}}

\newcommand{\bz}{\mathbf{z}}